\documentclass{article}
\usepackage[utf8]{inputenc}

\usepackage[margin=1in]{geometry}
\usepackage{amsmath}
\usepackage{enumerate}
\usepackage{amsthm}
\usepackage{amssymb}
\usepackage{graphicx}
\usepackage{tikz}
\usepackage{pgfplots}
\usetikzlibrary{arrows}
\usepackage{caption}
\usepackage{subcaption}
\usepackage{bbm}
\usepackage{enumitem}
\usepackage[colorlinks=true]{hyperref}
\usepackage{cleveref}
\usepackage{float}
\usepackage{thmtools, thm-restate}
\usepackage{algorithm}%
\usepackage{algorithmicx}%
\usepackage{algpseudocode}%

\usepackage[T1]{fontenc}

\tikzstyle{vertex}=[circle, draw, fill, inner sep=0pt, minimum size=5pt]

\newcommand\abs[1]{\left\lvert{#1}\right\rvert}
\newcommand\ceil[1]{\left\lceil{#1}\right\rceil}

\newcommand\norm[1]{\left\lVert{#1}\right\rVert}
\newcommand\wtd[1]{\widetilde{#1}}

\DeclareMathOperator{\gap}{\mathrm{gap}}

\DeclareMathOperator{\Null}{\mathrm{Null}}

\def\Am{\mathbf{A}}
\def\Bm{\mathbf{B}}
\def\Lm{\mathbf{L}}

\def\R{\mathbb{R}}

\def\sE{\mathcal{E}}
\def\sB{\mathcal{B}}

\def\trans{{^\top}}
\def\E{\mathbb{E}}

\def\bv{\mathbf{b}}
\def\xv{\mathbf{x}}
\def\fv{\mathbf{f}}

\def\rv{\mathbf{r}}

\def\uv{\mathbf{u}}
\def\vv{\mathbf{v}}
\def\yv{\mathbf{y}}
\def\zv{\mathbf{z}}

\def\one{\mathbbm{1}}

\def\poly{\mathrm{poly}}
\def\st{\mathrm{st}}

\def\ALGNAME{\textsf{Dual KOSZ}}
\def\OALGNAME{\textsf{KOSZ}}
\mathcode`l="8000
\begingroup
\makeatletter
\lccode`\~=`\l
\DeclareMathSymbol{\lsb@l}{\mathalpha}{letters}{`l}
\lowercase{\gdef~{\ifnum\the\mathgroup=\m@ne \ell \else \lsb@l \fi}}%
\endgroup

\newtheorem*{claim*}{Claim}

\newtheorem{lemma}{Lemma}
\newtheorem*{lemma*}{Lemma}

\newtheorem{prop}{Proposition}
\newtheorem*{prop*}{Proposition}
\newtheorem{theorem}{Theorem}
\newtheorem*{theorem*}{Theorem}
\newtheorem{defn}{Definition}
\newtheorem*{defn*}{Definition}

\newtheorem*{convention*}{Convention}

\newtheorem{fact}{Fact}

\newcommand*\samethanks[1][\value{footnote}]{\footnotemark[#1]}

\theoremstyle{plain}

\newenvironment{remark}{\noindent{\bf Remark}\hspace*{1em}}{\bigskip}
\usepackage{multicol}

\allowdisplaybreaks

\begin{document}

\title{A Combinatorial Cut-Toggling Algorithm for Solving Laplacian Linear Systems}


\author{Monika Henzinger\thanks{Institute of Science and Technology Austria. {\tt monika.henzinger@ista.ac.at}} \and  Billy Jin\thanks{School of Operations Research and Information Engineering, Cornell University. {\tt \{bzj3, davidpwilliamson\}@cornell.edu}.} \and Richard Peng\thanks{Cheriton School of Computer Science, University of Waterloo.  Email: {\tt y5peng@uwaterloo.ca}. } \and David P.\ Williamson\samethanks[2]}








\maketitle

\begin{abstract}
Over the last two decades, a significant line of work in theoretical algorithms has made progress in solving linear systems of the form $\mathbf{L}\mathbf{x} = \mathbf{b}$, where $\mathbf{L}$ is the Laplacian matrix of a weighted graph with weights $w(i,j)>0$ on the edges.  The solution $\mathbf{x}$ of the linear system can be interpreted as the potentials of an electrical flow in which the resistance on edge $(i,j)$ is $1/w(i,j)$.  Kelner, Orrechia, Sidford, and Zhu \cite{KOSZ13} give a combinatorial, near-linear time algorithm that maintains the Kirchoff Current Law, and gradually enforces the Kirchoff Potential Law by updating flows around cycles ({\it cycle toggling}).  

In this paper, we consider a dual version of the algorithm that maintains the Kirchoff Potential Law, and gradually enforces the Kirchoff Current Law by {\it cut toggling}: each iteration updates all potentials on one side of a fundamental cut of a spanning tree by the same amount. We prove that this dual algorithm also runs in a near-linear number of iterations.  

We show, however, that if we abstract cut toggling  as a natural data structure problem, this problem can be reduced to the online vector-matrix-vector problem (OMv), which has been conjectured to be difficult for dynamic algorithms \cite{HKNS15}.  The conjecture implies that the data structure does not have an $O(n^{1-\epsilon})$ time algorithm for any $\epsilon > 0$, and thus a straightforward implementation of the cut-toggling algorithm requires essentially linear time per iteration. 

To  circumvent the lower bound, we batch update steps, and perform them simultaneously instead of sequentially.  An appropriate choice of batching leads to an $\widetilde{O}(m^{1.5})$ time cut-toggling algorithm for solving Laplacian systems.  Furthermore, we show that if we sparsify the graph and call our algorithm recursively on the Laplacian system implied by batching and sparsifying, we can reduce the running time to $O(m^{1 + \epsilon})$ for any $\epsilon > 0$. Thus, the dual cut-toggling algorithm can achieve (almost) the same running time as its primal cycle-toggling counterpart. 
\end{abstract}


\section{Introduction}

Over the last two decades, a significant line of work in theoretical algorithms has made progress in solving linear systems of the form $\mathbf{L}\mathbf{x} = \mathbf{b}$, where $\mathbf{L}$ is the Laplacian matrix of a weighted graph with weights $w(i,j)>0$ on the edges. Starting with the work of Spielman and Teng \cite{SpielmanT04}, researchers have devised a number of algorithms that run in near-linear time in the number of edges of the graph (corresponding to the number of non-zeros in the matrix $\mathbf{L}$).  The solution $\mathbf{x}$ of the linear system can be interpreted as the potentials of an electrical flow in which the resistance of each edge is $r(i,j)=1/w(i,j)$, and the current supplied to each node $i$ is $b(i)$.  There have been many nice algorithmic ideas introduced using this interpretation of the linear system, as well as many applications of the fast algorithms for solving this system to other flow problems, such as the maximum flow problem
\cite{CKLPPS_max,ChristianoKMST11,Kathuria2022Unit,Madry16,Gao2021fully}.

Since the initial work of Spielman and Teng, a number of different near-linear time algorithms have been proposed.  In this paper, we wish to focus on a particular simple, combinatorial algorithm by Kelner, Orrechia, Sidford, and Zhu \cite{KOSZ13}, hereafter referred to as the \OALGNAME\ algorithm; this algorithm has been the subject of several implementation studies \cite{BomanDG15,BomanDG16,DeweeseGMPXX16,HoskeLMW16}.  The \OALGNAME\ algorithm uses the idea of an \emph{electrical flow} in solving the linear system $\mathbf{L}\mathbf{x} = \mathbf{b}$.  An electrical flow $\mathbf{f}$ is one that obeys the flow conservation constraints at each node $i$, saying that the net flow out of $i$ is $b(i)$ (sometimes known in this context as the Kirchoff Current Law, or KCL), and Ohm's law, which says that there exists a vector $\mathbf{x}$ such that the flow from $i$ to $j$, $f(i,j)$, equals $(x(i) - x(j))/r(i,j)$.  There exists a flow $\mathbf{f}$ that satisfies these two properties, and the corresponding potential vector $\mathbf{x}$ solves $\mathbf{L} \mathbf{x} = \mathbf{b}$.  Ohm's Law is known to be equivalent to the Kirchoff Potential Law (KPL), which says that the flow $\mathbf{f}$ satisfies the property that around any directed cycle $C$, $\sum_{(i,j) \in C} r(i,j) f(i,j) = 0$.  Given KPL, potentials satisfying Ohm's law can be constructed by picking any spanning tree $T$ rooted at a vertex $r$, setting $x(r)$ to 0, and $x(k)$ to the sum of $f(i,j)r(i,j)$ on the path in $T$ from $k$ to $r$; these potentials are known as {\em tree-induced potentials}.

The \OALGNAME\ algorithm starts by picking a spanning tree $T$; for the running time of the algorithm, it is important that the tree has low {\em stretch}; its definition is otherwise not crucial to the description of the algorithm.
The algorithm starts by constructing a flow $\mathbf{f}$ that satisfies flow conservation using only the edges in $T$.   For a near-linear number of iterations, the algorithm picks at random a non-tree edge $(i,j)$ and considers the fundamental cycle $C$ closed by adding $(i,j)$ to $T$; it then alters the flow $\mathbf{f}$ along $C$ to satisfy KPL (and such that KCL continues to be satisfied).  By picking $(i,j)$ with the appropriate probability, Kelner et al.\ show that the energy of the resulting flow decreases by a factor $1 - \frac{1}{\tau}$ in expectation, where $\tau$ is a parameter related to the stretch. The algorithm then returns the tree-induced potentials $\mathbf{x}$ associated with $T$ and the flow $\mathbf{f}$.  Kelner et al.\ \cite{KOSZ13} show that the resulting potentials are close to the potentials of the associated electrical flow for the graph.  The \OALGNAME\ algorithm has the pleasing properties that it is easy to understand both the algorithm and the analysis, and it is also close in spirit to known network flow algorithms; in particular, it resembles the primal network simplex algorithm for the minimum-cost flow problem.

The primal network simplex algorithm for the minimum-cost flow problem has a natural dual algorithm in the dual network simplex algorithm, in which node potentials are altered on one side of a fundamental cut of a tree (the cut induced by removing an edge in the tree).  Similarly, polynomial-time cycle-canceling algorithms for minimum-cost flow (e.g.\ Goldberg and Tarjan \cite{GoldbergT89}) have natural dual analogs of polynomial-time cut-cancelling algorithms (e.g.\ Ervolina and McCormick \cite{ErvolinaM93}).  We refer to the first type of algorithm as a {\em cycle-toggling} algorithm, and the dual analog as a {\em cut-toggling} algorithm. 
Thus, the \OALGNAME~algorithm is a cycle-toggling algorithm for solving Laplacian linear systems. However, no corresponding cut-toggling algorithm exists in the literature, leading immediately to the following question:
\begin{center}
{\emph{Does there exist a cut-toggling algorithm for solving Laplacian linear systems / computing near-minimum energy flows, and how efficiently can it be implemented?}}
\end{center}
\subsection{Our Contributions}
We propose a dual analog of the \OALGNAME\ algorithm which performs cut-toggling rather than cycle-toggling. We refer to this algorithm as \ALGNAME,
and we show it converges in a nearly-linear number of cut-toggling steps. Thus, \ALGNAME\ would be a nearly-linear time algorithm if each cut-toggling operation could be implemented to run in polylogarithmic time.

Our next contribution is to show that the natural data structure abstraction of this
cut-toggling process can be reduced to the
online matrix-vector (OMv) problem,
which has been conjectured to be hard~\cite{HKNS15}.
This implies that it is unlikely for there to be a black-box
data structure that implements a single cut-toggling operation
in sublinear time, unlike cycle toggling.

This result initially seems to present an insurmountable difficulty to obtaining a nearly-linear time algorithm. However, we show that we can exploit the offline nature of the
cut toggles, obtaining an exact data structure
for computing a sequence of $K$ cut-toggles in
$O(K\sqrt{m})$ time total, which yields an algorithm that runs in $\tilde{O}(m^{1.5})$ time overall.  Interestingly, Boman, Deweese, and Gilbert \cite{BomanDG16} explored an implementation of \OALGNAME\ that also batched its cycle-toggling updates by looking for collections of edge-disjoint cycles.

We further show that by incorporating sparsification and its associated
approximations, we can reduce the running time to almost-linear,
which means the cut toggling algorithm can still be implemented
in almost-linear time.

Our result demonstrates that graph optimization algorithms
and dynamic graph data structures can -- and sometimes need to -- interact in more intricate fashion than optimization in the outer loop and black-box data structure in the inner loop.

The remainder of the paper is structured as follows.  In Section \ref{sec:technical}, we give a high-level overview of \ALGNAME\ and our implementation ideas to obtain almost linear time.  Section \ref{sec:notation} describes some notation and concepts we will use.  Section \ref{sec:alg} gives the \ALGNAME\ in detail and shows that it can be implemented in a near-linear number of cut-toggling iterations, with each iteration running in linear time.  In Section \ref{sec:ds}, we abstract the problem of toggling a cut to a data structure problem, and show that given the OMv conjecture of \cite{HKNS15}, we cannot implement the operations needed in sublinear time.  In Section \ref{sec:speedup}, we show how to overcome this difficulty by batching the cut-toggle operations, and further speed up the algorithm through sparsification and recursion. 

\subsection{Technical Overview} \label{sec:technical}
As the \OALGNAME\ algorithm maintains a flow $\mathbf{f}$
and implicitly tree-induced potentials $\mathbf{x}$,
its natural dual is to maintain potentials $\mathbf{x}$,
which implicitly define a flow $\mathbf{f}$.\footnote{To make matters somewhat confusing, the Laplacian solver literature treats the space of potentials as {\em primal} due to its origins in numerical analysis.}
The \ALGNAME\ algorithm starts by choosing a low-stretch spanning tree $T$. It maintains a set of potentials $\mathbf{x}$ (initially zero), and the corresponding (infeasible) flow $\mathbf{f}$ implied by Ohm's Law.  In each iteration, we sample a fundamental cut $S$ of the tree $T$ and perform a cut-toggling update so that the net flow leaving $S$ is $\sum_{i \in S} b(i)$, as required in every feasible flow.  Following arguments dual to those made in Kelner et al.\ we show that this algorithm also performs a near-linear number of iterations in order to find a near-optimal set of potentials $\xv$  and flow $\fv$. 

\begin{restatable}{theorem}{dualkosz}
\label{thm:dual_kosz}
Let $\tau$ be the total stretch of $T$. After $K = \tau \ln(\frac\tau\epsilon)$ iterations, \ALGNAME\ returns $\xv^K \in \R^V$ and $\fv^K \in \R^{\vec{E}}$ such that $\E\norm{\mathbf{x}^* - \mathbf{x}^K}_\mathbf{L}^2 \leq \frac{\epsilon}{\tau}\norm{\mathbf{x}^*}_\mathbf{L}^2$ and $\E[\sE(\fv^K)] \leq (1+\epsilon) \sE(\fv^*)$, for $\fv^*$ and $\xv^*$ optimal primal and dual solutions respectively.
\end{restatable}
Here $\|\yv\|_{\Lm} = \sqrt{\yv^\trans \Lm \yv}$, and $\sE(\fv)$ is the energy of flow $\fv$.

However, unlike Kelner et al., we cannot show that each individual cut-toggling update can be made to run in polylogarithmic time. If we abstract the desired cut-toggling update step as a natural data structure problem,  we show that such a data structure cannot be implemented in $O(n^{1-\epsilon})$ time for any $\epsilon > 0$ given a conjecture about the {\em online matrix-vector multiplication problem (OMv)} made by Henzinger, Krinninger, Nanongkai and Saranurak \cite{HKNS15}. They have conjectured that this problem does not have any algorithm that can carry out an online sequence of $n$ Boolean matrix-vector multiplications in time $O(n^{3-\epsilon})$, and show that if the conjecture is false, then various long-standing dynamic graph problems will have faster algorithms.  We show that a single Boolean matrix-vector multiply can be carried out as a sequence of $O(n)$ operations of our desired data structure.  Given the conjecture, then, we cannot implement the data structure operations in $O(n^{1-\epsilon})$ time. Thus there is not a straightforward  near-linear time version of the \ALGNAME\  algorithm.\footnote{In a personal communication, Sherman \cite{Sherman17} said he also had worked out a dual version of the \OALGNAME\  algorithm, but was unable to solve the data structure problem for the updates to potentials.  Our result explains why this might be difficult to do.}

In a bit more detail, the data structure we define is given as input an undirected graph with a spanning tree $T$, edge resistances and a supply $b(v)$ and a potential $\mathsf{value}(v)$ at every vertex $v$. The data structure supports two operations: 1) additively update the value of all the vertices in a given subtree of $T$, and 2) query the value of the flow induced by the potential values across any given fundamental cut of $T$. We show that even if the data structure can only decide whether the value of the flow is non-zero or not, it would refute the OMV conjecture. Therefore, we obtain a lower bound for this data structure conditional on the OMv conjecture~\cite{HKNS15}. This even holds if all edge resistances are 1 and all $b(v)$ are 0.

Nevertheless, we circumvent this data structural lower bound  by exploiting the fact that the sequence of cuts to be updated can be sampled in advance and, thus, the updates can be batched, circumventing the ``online'' (or ``sequential'') requirement in OMv. This is possible because both the spanning tree $T$ and the probability distribution over cuts of $T$ are fixed at the beginning of the algorithm. More precisely, denote the number of iterations of \ALGNAME\ by $K$ (which is $\widetilde{O}(m)$). 
Instead of sampling the fundamental cuts one at a time, consider sampling the next $l$ cuts that need to be updated for some $l \ll K$. In each ``block" of size $l \ll K$, we contract all the edges of $T$ that do not correspond to one of the $l$ fundamental cuts to be updated. In this way, we work with a contracted tree of size $O(l)$ in each block (instead of the full tree, which has size $O(n)$). This makes the updates faster. However, the price we pay is that at the end of each block, we need to propagate the updates we made (which were on the contracted tree), back to the entire tree. Overall, we show that each block takes $O(l^2 + m)$ time. Since there are $\widetilde{O}(\frac{m}{l})$ blocks, the total runtime is $\wtd{O}(ml + \frac{m^2}{l})$.
Choosing $l=\sqrt{m}$ thus gives a $\tilde{O}(m^{1.5})$ time algorithm.  

By augmenting the batching idea with sparsification and recursion, one can further improve the running time of \ALGNAME\  to $\widetilde{O}(m^{1+\delta})$ for any $\delta > 0$.  To do this, observe that $l$ cut-toggling updates effectively break the spanning tree into $l+1$ components. After contracting the components to get a graph $H$ with $l+1$ vertices, we can show that solving an appropriate Laplacian system on $H$ gives a single update step that makes at least as much progress as the sequence of $l$ updates performed by the straightforward unbatched algorithm. A natural approach is to solve this Laplacian system by recursively calling the algorithm. However, this by itself does not give an improved running time. Instead, we first spectrally sparsify $H$ and then call the algorithm recursively to solve the \emph{sparsified} Laplacian system. Here we use a combinatorial  spectral sparsifier~\cite{kyng17sparsify} because it does not
require calling Laplacian solvers as a subroutine (e.g. \cite{BatsonSST13}).
By carefully analyzing the error incurred by sparsification, we are able to show that the update step using sparsification  makes about as much progress as the update step without sparsification. The total running time of the recursive algorithm is then obtained by bounding the time taken at each layer of the recursion tree.

\begin{restatable}{theorem}{fastest}
\label{thm:time}
For any $\delta \in (0,1)$ and $\epsilon > 0$, \ALGNAME\ with batching, sparsification, and recursion finds $\xv$ with $\E\norm{\mathbf{x}^* - \mathbf{x}}_\mathbf{L}^2 \leq \epsilon\norm{\mathbf{x}^*}_\mathbf{L}^2$ in 
$O(m^{1+\delta}(\log n)^{\frac{8}{\delta} }(\log \frac{1}{\epsilon})^{\frac{1}{\delta}})$ time.
\end{restatable}

\section{Notation and Problem Statement} \label{sec:notation}
In this section, we give some notation and define concepts we will use in our algorithm.  We are given as input an undirected graph $G = (V, E)$, with positive resistances $\mathbf{r} \in \R_+^E$. Although our graph is undirected, it will be helpful for us notationally to direct the edges. To that end, fix an arbitrary orientation of the edges, and denote the set of directed edges by $\vec{E}$.

In addition to $G$ and the resistances $\mathbf{r}$, we are given a \emph{supply vector} $\mathbf{b} \in \R^V$. 

The \emph{Laplacian matrix} of $G$ with respect to the resistances $\mathbf{r}$ is the matrix $\mathbf{L} \in \R^{V \times V}$ defined by 
$$\mathbf{L} = \sum_{ij \in E} \frac{1}{r(i,j)}(\mathbf{e_i} - \mathbf{e_j})(\mathbf{e_i} - \mathbf{e_j})\trans,$$
where $\mathbf{e_i}$ is the $i$th unit basis vector. We note then that $\mathbf{x}\trans \mathbf{L} \mathbf{x} = \sum_{(i,j) \in E} \frac{1}{r(i,j)}(x(i)-x(j))^2$ for any vector $\mathbf{x}$.

Our goal is to solve the system of linear equations $\mathbf{L}\mathbf{x} = \mathbf{b}$ for $\mathbf{x}$. However, we will not be able to solve $\mathbf{L}\mathbf{x} = \mathbf{b}$ exactly, so we will solve it approximately instead. It is usual to measure the quality of a solution $\mathbf{x}$ in terms of the \emph{matrix norm} induced by $\mathbf{L}$. In other words, if $\mathbf{x}$ is the vector of potentials returned by our algorithm and $\mathbf{x}^*$ is an actual solution to $\mathbf{L}\mathbf{x}^* = \mathbf{b}$, then the \emph{error} of our solution $\mathbf{x}$ is
$$\norm{\mathbf{x}^* - \mathbf{x}}_{\Lm}^2 := \left(\mathbf{x}^* - \mathbf{x}\right)\trans \mathbf{L} \left(\mathbf{x}^* - \mathbf{x}\right).$$
Hence, our objective is to find $\mathbf{x} \in \R^V$ that minimizes $\norm{\mathbf{x}^* - \mathbf{x}}_{\Lm}^2$. A precise statement of this is given below.

\begin{center}
\framebox{\textbf{Goal:} Given $\epsilon > 0$, find potentials $\mathbf{x} \in \R^V$ that satisfy $\norm{\mathbf{x}^* - \mathbf{x}}_\mathbf{L}^2 \leq \epsilon\norm{\mathbf{x}^*}_\mathbf{L}^2$.}
\end{center}
Of course, the algorithm does not know the actual solution $\mathbf{x}^*$. The place where $\mathbf{x}^*$ appears is in the analysis.

Equations of the form $\mathbf{L}\mathbf{x} = \mathbf{b}$, where $\mathbf{L}$ is the Laplacian matrix of a graph, are called Laplacian systems, and are found in a wide variety of applications in computer science and other fields. When interpreted as an optimization problem, solving a Laplacian linear system has the following nice interpretation: it is the dual of the problem of finding an electrical flow. The primal problem below is that of finding an electrical flow; here $\Am$ is the vertex-arc incidence matrix of $(V, \vec{E})$. The optimal solution to the primal is called the \emph{electrical flow} in $G$ defined by the resistances $\rv$. The dual problem is equivalent to solving $\Lm\xv = \bv$;  this fact can be seen by setting the gradient of the dual objective to equal 0 since the dual is concave.

\begin{align*}
(P) \quad &\min  \quad \frac12 \sum_{(i,j)\in\vec{E}} r(i,j){f(i,j)}^2  
\qquad\qquad (D)  &\max \quad \bv\trans\xv - \frac12 \xv\trans \Lm \xv \\
&\text{s.t.} \quad \Am\fv = \bv 
&\text{s.t.} \quad \xv \in \R^{V}~\text{unconstrained} 
\end{align*}

Let $\sE(\fv)$ denote the primal objective and $\sB(\xv)$ denote the dual objective. The primal objective is sometimes referred to as the {\em energy} of the electrical flow $\fv$.
Also, let $\fv^*$ denote the optimal primal solution and let $\xv^*$ denote an optimal dual solution. Note that there are infinitely many dual solutions, because the dual objective is invariant under adding a constant to every component of $\xv$.  
By strong duality (note that Slater's condition holds), $\sE(\fv^*) = \sB(\xv^*)$. Moreover,
the KKT conditions give a primal-dual characterization of optimality.
\begin{fact}[KKT Conditions for Electrical Flow]
\label{fact:kkt}
Consider $\fv \in \R^{\vec{E}}$ and $\xv \in \R^V$. Then $\fv$ is optimal for the primal and $\xv$ is optimal for the dual if and only if the following conditions hold:
\begin{enumerate}
    \item $\fv$ is a feasible $\bv$-flow;
    \item (Ohm's Law) \label{ohm} $f(i, j) = \frac{x(i) - x(j)}{r(i,j)}$ for all $(i, j) \in \vec{E}$.
\end{enumerate}
\end{fact}
Thus solving $\mathbf{L}\mathbf{x} = \mathbf{b}$ and finding the electrical $\mathbf{b}$-flow in $G$ are equivalent: Given a solution $\mathbf{x}$ to $\mathbf{L}\mathbf{x} = \mathbf{b}$, we can calculate the electrical flow $\mathbf{f}$ using $f(i,j) = \frac{x(i) - x(j)}{r(i, j)}$. On the other hand, given the electrical flow $\mathbf{f}$, we can recover corresponding potentials $\mathbf{x}$ by setting $x(v) = 0$ for some arbitrary vertex $v$, and using the equation $f(i,j) = \frac{x(i) - x(j)}{r(i,j)}$ to solve for the potentials on every other vertex. 

A $\mathbf{b}$-flow satisfies Ohm's Law if and only if it satisfies the Kirchoff Potential Law (KPL): KPL states that for every directed cycle $C$, $\sum_{(i,j) \in C} f(i,j)r(i,j) = 0$.

Both the Kelner et al.\ algorithm and our algorithm use a low-stretch spanning tree $T$.  Given resistances $\mathbf{r}$, the {\em stretch} of a tree is defined as 
$$
    \st_T(G) = \sum_{(i, j) \in \vec{E}} \st_T(i, j)
    = \sum_{(i, j) \in \vec{E}} \frac{1}{r(i, j)}\sum_{(k, l) \in P(i,j)} r(k, l),
    $$
    where $P(i, j)$ is the unique path from $i$ to $j$ in $T$. 
We can find a spanning tree $T$ with stretch $\st_T(G)=O(m\log n \log\log n)$ in $O(m \log n \log\log n)$ time \cite{AN12}.

We use the notation $\one$ to stand for the vector of all 1s, and $\one_X$ to be the characteristic vector of a set $X$ that has 1s in the entries corresponding to the elements of $X$ and 0s elsewhere.

\section{A Cut-Toggling Algorithm for Solving Laplacian Linear Systems}
\label{sec:alg}
We present a cut-toggling algorithm for computing an approximate solution to $\Lm\xv = \bv$, and also an approximate minimum-energy $\bv$-flow. The goal of this section is to show that the cut-toggling algorithm converges in a near-linear number of iterations, and that each iteration runs in linear time. Later in \Cref{sec:speedup}, we will show how to speed up the algorithm to an almost-linear total running time. Since our algorithm is dual to the cycle-toggling algorithm of Kelner et al. \cite{KOSZ13} (which we call \OALGNAME\ in this paper), we will begin by describing the \OALGNAME\ algorithm.

The  \OALGNAME\ algorithm  works by maintaining a feasible $\bv$-flow $\fv$, and iteratively updates $\fv$  along cycles to satisfy Kirchkoff's Potential Law on the cycle.  It starts by choosing a spanning tree $T$ that has low stretch, and computes a $\mathbf{b}$-flow $\mathbf{f}^0$ that uses only edges in the tree $T$.  Then for a number of iterations $K$ that depends on the stretch of the tree, it chooses a non-tree edge $(i,j) \in E-T$ according to a probability distribution, and for the fundamental cycle closed by adding edge $(i,j)$ to $T$, it modifies the flow $\mathbf{f}$ so that Kirchoff's Potential Law is satisfied on the cycle.  The probability $P_{ij}$ that edge $(i,j)$ gets chosen is proportional to the total resistance around the cycle closed by $(i,j)$ divided by $r(i,j)$. Given the tree $T$  with root $r$ and the current flow $\mathbf{f}^t$ in iteration $t$, there is a standard way to define a set of potentials $\mathbf{x}^t$ (called the {\em tree-induced} or {\em tree-defined} potentials):
set $x(r)$ to 0, and $x(k)$ to the sum of $f(i,j)r(i,j)$ on the path in $T$ from $k$ to $r$.
We summarize \OALGNAME\  in Algorithm \ref{KOSZalg}.

\begin{algorithm}[t]
\begin{algorithmic}[1]
	\State Compute a  tree $T$ with low stretch with respect to resistances $\mathbf{r}$\;
	\State Find flow $\mathbf{f}^0$ in $T$ satisfying supplies $\mathbf{b}$\;
	\State Let $\mathbf{x}^0$ be tree-defined potentials for $\mathbf{f}^{0}$ with respect to tree $T$\;
	\For{$t \gets 1$ to $K$}
        \State Pick an $(i,j) \in E-T$ with probability $P_{ij}$
		\State Update $\mathbf{f}^{t-1}$ to satisfy KPL on the fundamental cycle closed by $(i,j)$\;
		\State Let $\mathbf{f}^t$ be resulting flow\;
		\State Let $\mathbf{x}^t$ be tree-defined potentials for $\mathbf{f}^t$\;
    \EndFor
	\State \Return{$\mathbf{f}^K, \mathbf{x}^K$}
	\caption{The \OALGNAME\  algorithm for solving $\mathbf{L}\mathbf{x} = \mathbf{b}$.}
	\label{KOSZalg}
	\end{algorithmic}
\end{algorithm}

Our algorithm, which we will call \ALGNAME, works by maintaining a set of potentials $\xv$. It iteratively samples cuts in the graph, updating potentials on one side of the cut to satisfy flow conservation across that cut.  Following \OALGNAME, we choose a spanning tree $T$ of low stretch.  Then for a number of iterations $K$ that depends on the stretch of tree $T$, we repeatedly sample a fundamental cut from the spanning tree (i.e.\ a cut induced by removing one of the tree edges). We update all of the potentials on one side of the cut by an amount $\Delta$ so that the amount of flow crossing the cut via Ohm's Law is what is required by the supply vector. We summarize \ALGNAME\ in Algorithm \ref{ouralg}. The main result of this section is a bound on the iteration complexity of \ALGNAME.

\dualkosz*

Next we give the algorithm in somewhat more detail.  Let $R(C) = (\sum_{(k,l) \in \delta(C)} \frac{1}{r(k,l)})^{-1}$ for $C \subset V$, where $\delta(C)$ is the set of edges with exactly one endpoint in $C$. Note that $R(C)$ has units of resistance. For every tree edge $(i,j)$, let  $C(i, j)$ be the set of vertices on one side of the fundamental cut defined by $(i, j)$, such that $i \in C(i, j)$ and $j \not\in C(i, j)$. We set up a probability distribution $P_{ij}$ on edges $(i,j)$ in the spanning tree $T$, where $P_{ij} \propto \frac{r(i,j)}{R(C(i,j))}$.  We initialize potentials $x^0(i)$ to 0 for all nodes $i \in V$. In each iteration, we sample edge $(i,j) \in T$ according to the probabilities $P_{ij}$.    Let $b(C) =\mathbf{b}\trans \one_C$ be the total supply of the nodes in $C$. Note that $b(C)$ is also the amount of flow that should be flowing out of $C$ in any feasible $\mathbf{b}$-flow. 

Let $f^{t}(C)$ be the total amount of flow going out of $C$ in the flow induced by $\mathbf{x}^{t}$. That is,
        $$f^{t}(C) = \sum_{\substack{ij \in E \\ i \in C, \, j \not\in C}} \frac{{x}^t(i) - {x}^t(j)}{r(i,j)}.$$
Note that $f^t(C)$ can be positive or negative. In any feasible $\mathbf{b}$-flow, the amount of flow leaving $C$ should be equal to $\mathbf{b}\trans \one_C = b(C)$. Hence, we define 
$\Delta^t = (b(C) - f^{t}(C))\cdot R(C).$
Observe that $\Delta^t$ is precisely the quantity by which we need to increase the potentials of every node in $C$ so that flow conservation is satisfied on $\delta(C)$. We then update the potentials, so that 
        $$
        p^{t+1}(v) = 
        \begin{cases}
        p^{t}(v) +\Delta^t, &\text{if $v \in C$,} \\
        p^{t}(v), &\text{if $v \not\in C$.}
        \end{cases}
        $$
After $K$ iterations, we return the final potentials $\mathbf{x}^{K}$. The last step is to convert $\xv^K$ to a feasible flow by taking a \emph{tree-defined flow} with respect to $T$: $f^K(i,j) = \frac{x^K(i)-x^K(j)}{r(i,j)}$ on all non-tree edges, and $\fv^K$ routes the unique flow on $T$ to make $\fv^K$ a feasible $\bv$-flow.

\begin{algorithm}[t]
\begin{algorithmic}[1]
	\State Compute a spanning tree $T$ with low stretch with respect to resistances $\rv$\;
	\State Set $x^0(i) = 0$ for all $i \in V$\;
	\For{$t \gets 1$ to $K$}
		\State Pick edge $(i,j) \in T$ with probability $P_{ij} \propto \frac{r(i,j)}{R(C(i,j)}$ and let $C = C(i,j)$\;
		\State $\Delta^t \gets (b(C)-f^t(C))\cdot R(C)$\;
		\State
        $
        x^{t+1}(v) \gets 
        \begin{cases}
        x^{t}(v) +\Delta^t, &\text{if $v \in C$,} \\
        x^{t}(v), &\text{if $v \not\in C$.}
        \end{cases}
        $
	\EndFor
	\State Let $\fv^K$ be the tree-defined flow with respect to $\xv^K$ and $T$\;
	\State \Return{$\mathbf{x}^K$, $\fv^K$}
	\caption{Algorithm \ALGNAME\ for solving $\mathbf{L}\mathbf{x} = \mathbf{b}$.}
	\label{ouralg}
	\end{algorithmic}
\end{algorithm}


\subsection{Analysis of \ALGNAME}
\label{sec:analysis}
Recall that $\sE(\fv) = \frac12\sum_{e \in E}{f(e)}^2$ and $\sB(\xv) = \bv^T\xv - \frac12\xv^T\Lm\xv$. 
By convex duality, we have $\sB(\xv) \leq \sE(\fv)$ for any $\xv\in \R^V$ and $\bv$-flow $\fv$. Moreover, $\xv$ maximizes $\sB(\xv)$ if and only if $\Lm\xv = \bv$. (See e.g.\ \cite[Lemma 8.9]{Williamson19}).
Thus solving the Laplacian system $\mathbf{L}\mathbf{x} = \mathbf{b}$ is equivalent to finding a vector of potentials that maximizes the dual objective. 
In what follows, we present the lemmas that form the bulk of the analysis. Their  proofs are deferred to the appendix.
These lemmas (and their proofs) are similar to their counterparts in Kelner et al.\ \cite{KOSZ13}, because everything that appears here is dual to what appears there.

First, we show that each iteration of the algorithm increases $\sB(\xv)$. 

\begin{restatable}{lemma}{enerincr}
\label{lem:energy_increase}
Let $\mathbf{x} \in \R^V$ be a vector of potentials and let $C \subset V$. Let $\mathbf{x}'$ be the potentials obtained from $\mathbf{x}$ as in the algorithm (that is, by adding $\Delta$ to the potential of every vertex in $C$ so that flow conservation is satisfied across $\delta(C)$). Then 
$$\sB(\mathbf{x}') - \sB(\mathbf{x}) = \frac{\Delta^2}{2R(C)}.$$
\end{restatable}

 The second ingredient in the analysis is to introduce an upper bound on how large the potential bound $\sB(\xv)$ can become. This will allow us to bound the number of iterations the algorithm takes.  

\begin{defn}[Gap]
Let $\mathbf{f}$ be a feasible $\mathbf{b}$-flow and let $\mathbf{x}$ be any vertex potentials. Define 
$$\gap(\mathbf{f}, \mathbf{x}) := \sE(\mathbf{f}) - \sB(\mathbf{x}) = \frac12\sum_{e \in E} r(e)f(e)^2 - \left(\mathbf{b}\trans \mathbf{x} -\frac12 \mathbf{x}\trans \mathbf{L}\mathbf{x}\right).$$
\end{defn}

This same notion of a gap was introduced in the analysis of the Kelner et al.\ algorithm, and was also used to bound the number of iterations of the algorithm.

The electrical flow $\mathbf{f}^*$ minimizes $\sE(\mathbf{f})$ over all $\mathbf{b}$-flows $\mathbf{f}$, and the corresponding vertex potentials $\mathbf{x}^*$ maximize $\sB(\mathbf{x}^*)$ over all vertex potentials $\mathbf{x}$. Moreover, $\sE(\mathbf{f}^*) = \sB(\mathbf{x}^*)$. Therefore, for any feasible flow $\mathbf{f}$, $\gap(\mathbf{f}, \mathbf{x})$ is an upper bound on optimality:
$$\gap(\mathbf{f}, \mathbf{x}) \geq \sB(\mathbf{x}^*) - \sB(\mathbf{x}).$$
The lemma below gives us another way to write $\gap(\mathbf{f}, \mathbf{x})$, and will be useful to us later. This relation is shown in Kelner et al.\  \cite[Lemma 4.4]{KOSZ13}.

\begin{restatable}{lemma}{gapp}
We have
$\gap(\mathbf{f}, \mathbf{x}) = \frac12\sum_{(i, j) \in \vec{E}} r(i, j) \left(f(i, j) - \frac{x(i)-x(j)}{r(i, j)}\right)^2.$
\end{restatable}

The analysis of Kelner et al. \cite{KOSZ13} relies on measuring progress in terms of the above-defined duality gap between primal flow energy and dual potential bound. The high-level idea of the analysis is that one can show that the duality gap decreases by a constant factor each iteration, which implies a linear convergence rate. In the analysis of their algorithm, they maintain a feasible $\mathbf{b}$-flow $\mathbf{f}$ at each iteration, and measure $\gap(\mathbf{f}, \mathbf{x})$ against corresponding tree-defined potentials $\mathbf{x}$.

One difference between their algorithm and ours is that we do not maintain a feasible $\mathbf{b}$-flow at each iteration. However, for $\gap(\mathbf{f}, \mathbf{x})$ to be a valid bound on distance to optimality, we need $\mathbf{f}$ to be a feasible $\mathbf{b}$-flow. To this end, we introduce the definition of ``tree-defined flow'' below.

\begin{defn}[Tree-defined flow]
\label{def:tree_def_flow}
Let $T$ be a spanning tree, $\mathbf{x} \in \R^V$ vertex potentials, and $\mathbf{b} \in \R^V$ satisfying $\one\trans \mathbf{b} = 0$ be a supply vector. The \textbf{tree-defined flow} with respect to $T$, $\mathbf{x}$ and $\mathbf{b}$ is the flow $\mathbf{f}_{T, \xv}$ defined by 
$$\mathbf{f}_{T, \xv}(i, j) = 
\frac{x(i) - x(j)}{r(i, j)} \quad \text{if $(i, j) \not\in T$},
$$
and for $(i, j) \in T$, $f_{T,\xv}(i, j)$ is the unique value such that the resulting $\mathbf{f}_{T, \xv}$ is a feasible $\mathbf{b}$-flow. That is, for $(i, j) \in T$, if $C = C(i, j)$ is the fundamental cut defined by $(i, j)$ and $b(C) = \mathbf{b}\trans \one_C$ is the amount of flow that should be flowing out of $C$ in a feasible $\mathbf{b}$-flow, then  
\begin{align*}
f_{T,\xv}(i, j)  = b(C) - \sum_{\substack{k \in C,\, l \not\in C \\ kl \in E - ij}} f_{T,\xv}(k, l) 
= b(C) - \sum_{\substack{k \in C,\, l \not\in C \\ kl \in E - ij}} \frac{x(k) - x(l)}{r(k, l)}.
\end{align*}
In other words, $\mathbf{f}_{T,\xv}$ is a potential-defined flow outside of the tree $T$, and routes the unique flow on $T$ to make it a feasible $\mathbf{b}$-flow. 
\end{defn}
The below lemma expresses $\gap(\fv_{T, \xv}, \mathbf{x})$ in a nice way. 
\begin{restatable}{lemma}{gaplem}
\label{lem:gap}
Let $T$ be a spanning tree, $\mathbf{x}$ vertex potentials, and $\mathbf{b}$ a supply vector. Let $\fv_{T, \xv}$ be the associated tree-defined flow. Then 
$$\gap(\fv_{T, \xv}, \mathbf{x}) = \frac12 \sum_{(i, j) \in T} r(i, j) \cdot \frac{\Delta(C(i, j))^2}{R(C(i, j))^2}.$$
\end{restatable}

Suppose we have a probability distribution $(P_{ij}: (i, j) \in T)$ on the edges in $T$. If the algorithm samples an edge $(i, j) \in T$ from this distribution, then by \Cref{lem:energy_increase} the expected increase in the dual objective is
$$
\E[\sB(\mathbf{x}')] - \sB(\mathbf{x}) =\frac12 \sum_{(i,j) \in T} P_{ij}\cdot\Delta(C(i, j))^2/R(C(i, j)).
$$
We want to set the $P_{ij}$ to cancel terms appropriately so that the right-hand side is a multiple of the gap. Looking at Lemma \ref{lem:gap}, we see that an appropriate choice is to set
$$P_{ij} := \frac{1}{\tau} \cdot \frac{r(i, j)}{R(C(i, j))},$$
where $\tau := \sum_{(i, j) \in T} \frac{r(i, j)}{R(C(i, j))}$ is the normalizing constant. For this choice of probabilities,
$$\E[\sB(\mathbf{x}')] - \sB(\mathbf{x}) 
= \frac{1}{2\tau} \sum_{(i, j) \in T} r(i,j)\cdot \frac{\Delta(C(i, j))^2}{R(C(i, j))^2} 
= \frac{1}{\tau}\gap(\fv_{T, \xv}, \mathbf{x}),$$
where $\fv_{T, \xv}$ is the tree-defined flow associated with potentials $\mathbf{x}$. As a consequence, we have:

\begin{restatable}{lemma}{gapdecr}
\label{lem:gap_decreases}
If each iteration of the algorithm samples an edge $(i, j) \in T$ according to the probabilities $P_{ij} = \frac{1}{\tau} \cdot \frac{r(i, j)}{R(C(i, j))}$, then we have
$$\sB(\mathbf{x}^*) - \E[\sB(\mathbf{x}^{t+1})]
\leq
\left(1 - \frac1\tau\right)
\left(\sB(\mathbf{x}^*) - \sB(\mathbf{x}^t)\right).$$
\end{restatable}

\begin{restatable}{corollary}{finalgap}
\label{cor:final_gap}
After $K = \tau\ln(\frac{1}{\epsilon})$ iterations, $\sB(\mathbf{x}^*) - \E[\sB(\mathbf{x}^K)] \leq \epsilon \cdot \sB(\mathbf{x}^*)$.
\end{restatable}

We now use the previous lemmas to bound the number of iterations \ALGNAME\ takes. Lemma \ref{lem:gap_decreases} shows that the quantity $\sB(\mathbf{x}^*) - \sB(\mathbf{x}^t)$ decreases multiplicatively by $(1 - \frac1\tau)$ each iteration. Thus, a smaller value of $\tau$ gives faster progress. Moreover, it is not difficult to show that $\tau = \st_T(G, \rv)$ (see \Cref{lem:stretch} in the Appendix), which is why $T$ is chosen to be a low-stretch spanning tree.

We also need to argue that rounding $\xv^K$ to $\fv^K$ via a tree-defined flow preserves approximate optimality. One can show that for any distribution over $\xv$ such that $\E_{\xv}[\sB(\xv)] \geq (1-\frac{\epsilon}{\tau})\sB(\xv^*)$, we have $\E_{\xv}[\sE(\fv_{T, \xv})] \leq (1+\epsilon)\sE(\fv^*)$.
    Combining everything together, we conclude: 
    \dualkosz*
    
     
    
    We end this section with a na\"{i}ve bound on the running time of \ALGNAME. 
    
    \begin{restatable}{lemma}{runtime}
    \label{lem:runtime}
    \ALGNAME\ can be implemented to run in $\wtd{O}(mn\log\frac{1}{\epsilon})$ time. 
    \end{restatable}
    
    \begin{proof}
    We can find a spanning tree $T$ with total stretch $\tau=O(m\log n \log\log n)$ in $O(m \log n \log\log n)$ time \cite{AN12}.
    
    For concreteness, fix an arbitrary vertex to be the root of $T$, and direct all edges in $T$ towards the root. The set of fundamental cuts we consider will be the vertex sets of subtrees of $T$.  
    
    To compute $b(C)$ for these $n-1$ fundamental cuts $C$, we can work our way from the leaves up to the root. If $C = \{v\}$ is a leaf of $T$, then $b(C) = b(v)$. Otherwise, $C$ is a subtree rooted at $v$, and $b(C) = b(v) + \sum_{C'} b(C')$, where the sum is over the subtrees of $C$. Hence we can compute $b(C)$ for all fundamental cuts $C$ in $O(n)$ time.
    
    To compute $R(C)$ for the fundamental cuts $C$, we can maintain $n-1$ variables, one for each fundamental cut. The variable corresponding to cut $C$ will represent $R(C)^{-1} = \sum_{e \in \delta(C)} \frac{1}{r(e)}$, and the variables are initialized to 0. We then iterate through all the edges in the graph, and for each such edge $e$, add $\frac{1}{r(e)}$ to the value of each variable that represents a cut $C$ such that $e \in \delta(C)$. Although this naive implementation takes $O(mn)$ time in the worst-case, it is possible to improve this running time to $O(m\log n)$ using link-cut trees \cite{ST83}. One can also achieve this running time using the same data structure as the one used in \cite{KOSZ13}. 
    
    The last part of the running time is the time it takes to run a single iteration of the algorithm.
    In each iteration of the algorithm, we need to compute $\Delta = (b(C) - f(C))\cdot R(C)$, where $C$ is the fundamental cut selected at that iteration. In the above two paragraphs, we described how to precompute the values of $b(C)$ and $R(C)$ for every fundamental cut $C$; note that these values are fixed at the beginning and do not change during the course of the algorithm. Hence, it remains to compute $f(C)$. One way to compute $f(C)$ is to simply iterate over all the edges in $\delta(C)$, summing each edge's contribution to $f(C)$. This takes time proportional to $\abs{\delta(C)}$, which could be $O(m)$ in the worst case. We can get this down to $O(n)$ per iteration by maintaining the values of $f(C)$ for every fundamental cut $C$, and updating these values each time the algorithm updates potentials. Since there are $n-1$ cuts, to do this in $O(n)$ time requires us to be able to update $f(C)$ for a single cut in $O(1)$ time. To do this, we can precompute an $(n-1) \times (n-1)$ table with a row/column for each fundamental cut, where the $(C_1, C_2)$ entry is the amount by which the flow out of $C_2$ increases if we add 1 to the potential of every node in $C_1$. Let $H(C_1, C_2)$ denote this value. With this table, updating the value of $f(C)$ after a potential update step essentially reduces to a single table lookup, which takes $O(1)$ time.
    
    Finally, note that one can construct the $H(C_1, C_2)$ table in $O(n^2)$ time using results from \cite{Karger00}. In the language of Definitions 5.3 and 5.5 in that paper, we are trying to compute $C(v^\downarrow, w^\downarrow)$ for all vertices $v,w$ in the tree, where the edge weights are the reciprocals of the resistances. At the bottom of page 11, it states that the $n^2$ values $C(v, w^\downarrow)$ can be computed in $O(n^2)$ time. At the top of Page 12, it then says that we get the values of $C(v^\downarrow, w^\downarrow)$ using $n$ treefix sums. (Each treefix sum is the procedure described in Lemma 5.8, and takes $O(n)$ time.))
    
    
    To summarize, we can run each iteration of the algorithm in $O(m)$ time, which can be reduced to $O(n)$ time if we precompute the $H(C_1, C_2)$ table, which incurs an overhead of $O(n^2)$ storage and $O(n^2)$ preprocessing time. 
    
    Suppose each iteration of the algorithm takes $O(I)$ time, and the algorithm uses $O(L)$ preprocessing time (not including the time needed to compute the low-stretch spanning tree). Then the total running time of the algorithm is $O(L + I \tau\ln(\frac{\tau}{\epsilon}) + m \log n \log\log n) = O(L + mI\log n \log\log n\log\frac{\tau}{\epsilon})$. 
    
    If we use the version which uses $O(n^2)$ preprocessing time and $O(n)$ time per iteration, then $L = O(n^2)$ and $I = O(n)$. This gives the running time of \ALGNAME\ to be $O(mn\log n \log\log n\log\frac{\tau}{\epsilon})$. 
    \end{proof}
    
     In Section \ref{sec:ds}, we argue that given a natural abstraction of the data structure problem we use in computing $f(C)$ and updating potentials, it appears unlikely that we can implement each iteration in $o(n^{1-\epsilon})$ time, if each iteration is to be processed one-by-one in an online fashion.  In Section \ref{sec:speedup}, we show how to overcome this data structure lower bound by taking advantage of the fact that the sequence of updates that we perform can be generated in advance.

\section{Lower Bound on the Per-Iteration Complexity of the Algorithm}
\label{sec:ds}
Recall that each single iteration of \OALGNAME\  can be implemented in logarithmic time.
In this section we show that assuming the OMv conjecture (see below) each single iteration of \ALGNAME\  cannot be implemented in linear time. This implies that in order to speed up our algorithm we need to ``batch-up'' iterations, which is the approach we use in the next section. We first present a natural data structure, called the\textit{TreeFlow} data structure, such that each iteration of the algorithm requires only two operations of the  \textit{TreeFlow} data structure and then prove that assuming the OMv conjecture~\cite{HKNS15} it is impossible to implement the \textit{TreeFlow} data structure such that each operation of the data structure takes $O(n^{1-\epsilon})$ time.
To simplify the reduction we reduce from a closely related problem called the \textit{Online Vector-Matrix-Vector Multiplication Problem (OuMv)}. 

\begin{defn}[Online Vector-Matrix-Vector Multiplication Problem]
We are given a positive integer $n$, and a Boolean $n\times n$ matrix $\mathbf{M}$. At each time step $t =1, \ldots, n$, we are shown a pair of Boolean vectors $(\mathbf{u_t}, \mathbf{v_t})$, each of length $n$. Our task is to output $\mathbf{u_t^\top} \mathbf{M} \mathbf{v_t}$ using Boolean matrix-vector operations. Specifically, ``addition" is replaced by the \textsf{OR} operation, so that $0+0 = 0$, and $0+1 = 1+0 = 1+1 = 1$. Hence, $\mathbf{u_t^\top} \mathbf{M} \mathbf{v_t}$ is always either 0 or 1. 
\end{defn}

The OMv conjecture implies that no algorithm for the OuMv problem can do substantially better than naively multiplying $\mathbf{u_t^\top} \mathbf{M} \mathbf{v_t}$ at time step $t$. Specifically, it says the following:

\begin{lemma}[\cite{HKNS15}]
\label{conj:omv}
Let $\epsilon > 0$ be any constant.
Assuming the OMv conjecture,
there is no algorithm for the online vector-matrix-vector multiplication problem that uses preprocessing time $O(n^{3-\epsilon})$ and takes total time $O(n^{3-\epsilon})$  with error probability at most $1/3$ in the word-RAM model with $O(\log n)$ bit words. 

\end{lemma}

Thus we will reduce the OuMv problem to the \textit{TreeFlow} data structure such that
computing $\mathbf{u_t^\top} \mathbf{M} \mathbf{v_t}$ requires two operations in the \textit{TreeFlow} data structure. The lower bound then follows from Lemma~\ref{conj:omv}.

\subsection{The \textit{TreeFlow}  Data Structure}
\label{subsec:ds}
The \textit{TreeFlow} data structure is given as input (1) an undirected graph $G = (V, E)$ with $n = \abs{V}$, (2) a spanning tree $T$ of $G$ that is rooted at a fixed vertex $x$, (3) a value $r(u,v)$ for each edge $(u,v) \in E$ (representing the resistance of $(u, v)$), and (4) a value $b(v)$ for each vertex $v \in V$ (representing the supply at $v$). The quantities $r(u, v)$ and $b(v)$ are given at the beginning and will remain unchanged throughout the operations. For any set $C \subset V$, let $b(C) := \sum_{v \in C} b(v)$.

Furthermore,
each vertex $v$ has a non-negative {value}, denoted $\mathsf{value}(v)$, which  can be seen as the ``potential'' of $v$. It is initially 0 and can be modified.
For any set $C \subset V$  we define the \emph{flow out of} $C$ to be the quantity
$$f(C) := \sum_{(u,v) \in E, u \in C, v \not\in C} \left(\mathsf{value}(u) - \mathsf{value}(v)\right)/r(u,v).$$

The \textit{TreeFlow} data structure supports the following operations.
\begin{itemize}
    \item $\textsf{addvalue(\textbf{vertex} $v$, \textbf{real} $x$)}$: Add $x$ to the value of every vertex in the subtree of $T$ rooted at $v$. 
    \item $\textsf{findflow(\textbf{vertex} $v$)}$: Return $b(C) - f(C)$, where $C$ is the set of vertices in the subtree of $T$ rooted at $v$. 
\end{itemize}
The \emph{TreeFlow} data structure implements exactly the operations we require for each iteration of \ALGNAME: The \textsf{addvalue} operation allows us to update the potentials on a fundamental cut, and \textsf{findflow} computes $b(C) - f(C)$, thereby allowing us to compute $\Delta$ at each iteration.
Note that if all $b(v)$-values are zero, the \emph{TreeFlow} data structure simply returns $-f(C)$, which gives it its name.

We even show the lower bound  for a ``relaxed'' version defined as follows: In an \emph{$\alpha$-approximate \textit{TreeFlow} data structure} the
operation \textsf{addvalue} remains as above and the operation  \textsf{findflow($v$)} returns a value that is within a multiplicative
 factor $\alpha \ge 1$ (that can be a function of $n$) of the correct answer, i.e., a value between $(b(C) - f(C))/\alpha$ and $(b(C) - f(C))\cdot \alpha$. The hardness of approximation is interesting, because it turns out that even an approximation of this quantity is sufficient to obtain an algorithm for $\mathbf{L}\mathbf{x} = \mathbf{b}$. (Albeit with a convergence rate that deteriorates with the approximation factor.)

\begin{restatable}{lemma}{reduction}
\label{lem:reduction}
Let $\epsilon > 0$ be any constant and let $\alpha \ge 1$ be any value.
Assuming the OMv conjecture, no implementation of the $\alpha$-approximate \textit{TreeFlow} data structure exists that uses preprocessing time $O(n^{3-\epsilon})$ and where the two operations \textsf{addvalue} and \textsf{findflow} both take $O(n^{1-\epsilon})$ time,
 such that over a polynomial number of operations the error probability is at most $1/3$ in the word-RAM model with $O(\log n)$ bit words. This even holds if all $r(u,v)$ values are 1 and if all $b(v)$ are 0.
\end{restatable}

\section{Speeding Up \ALGNAME}
\label{sec:speedup}
We now show how to surmount the OMv lower bound by taking advantage of the fact that the sequence of updates that \ALGNAME\  performs can be generated in advance.
In \Cref{sec:batch}, we show that batching the updates yields a modification of the algorithm that runs in $\wtd{O}(m^{1.5})$ time. Then in \Cref{sec:recurse}, we use sparsification and recursion to further improve the runtime to $\wtd{O}(m^{1+\alpha})$ for any $\alpha > 0$. 

\subsection{A Faster Algorithm using Batching}
\label{sec:batch}
First, we show that it is possible to speed up the running time to $\wtd{O}(m^{1.5})$ time by batching the updates performed by \ALGNAME. In Lemma \ref{lem:runtime}, we showed that the algorithm can be implemented to run in time $\widetilde{O}(mn)$. (Here the tilde hides a factor of  $\log n \log\log n \log\frac{1}{\epsilon}$.) 
This running time essentially comes from $\wtd{O}(m)$ iterations, and $O(n)$ time per iteration.
Recall that each iteration of \ALGNAME\ involves sampling a fundamental cut $C$ of the low-stretch spanning tree $T$ from a fixed probability distribution $P$, and then adding a constant to the potential of every vertex in $C$ so that the resulting potential-defined flow satisfies flow conservation across $C$. 

The main idea of batching is as follows. Denote the number of iterations by $K$ (which is $\widetilde{O}(m)$). 
Instead of sampling the fundamental cuts one at a time, consider sampling the next $l$ cuts that need to be updated for some $l \ll K$.  We can perform this sampling in advance because both the tree $T$ and the probability distribution over cuts of $T$ are fixed over the entire course of the algorithm. In each ``block" of size $l \ll K$, we contract all the edges of $T$ that do not correspond to one of the $l$ fundamental cuts to be updated. In this way, we work with a contracted tree of size $O(l)$ in each block (instead of the full tree, which has size $O(n)$). This makes the updates faster. However, the price we pay is that at the end of each block, we need to propagate the updates we made (which were on the contracted tree), back to the entire tree. We will show that by choosing $l = \sqrt{m}$, we can balance this tradeoff and get an improved running time of $\wtd{O}(m^{1.5})$. Pseudocode for \ALGNAME\ with batching is given in Algorithm \ref{alg:faster}. Note that the correctness of this algorithm follows directly from the correctness of \ALGNAME: Algorithm \ref{alg:faster} samples cuts from exactly the same distribution as \ALGNAME, and if we fix the same sequence of cuts to be used by both algorithms, then the output of the two algorithms is identical.

\begin{algorithm}[ht!]
\begin{algorithmic}[1]
\caption{\ALGNAME\ with batching}
\label{alg:faster}
	\State Compute a  tree $T$ with low stretch with respect to resistances $\rv$.
	\State Compute $b(C)$ and $R(C)$ for all fundamental cuts $C$ of $T$.
	\State Set $\xv^0(i) = 0$ for all $i \in V$. Set $f(C) = 0$ for all fundamental cuts $C$ of $T$.
	\For{$t \gets 1$ to $\ceil{\frac{K}{l}}$}
	    \State Sample $l$ edges $(i_1,j_1), \ldots, (i_l, j_l)$ with replacement from $T$, according to the distribution $P$.
	    \State Contract all edges in $T$ that were not sampled in step 5. \\ Let $\wtd{G}$ be the resulting graph and $\wtd{T}$ be the resulting tree.
	    \State For each $1 \leq k\leq l$, let ${C}_k$ denote the fundamental cut in $T$ determined by edge $(i_k, j_k)$. Let $\wtd{C}_k$ denote the fundamental cut in $\wtd{T}$ determined by $(i_k, j_k)$.
	    \State $\yv(\tilde{v}) \gets 0$ for all $\tilde{v} \in V(\wtd{G})$.
	    \For{$k \gets 1$ to $l$}
	        \State Compute $\Delta_k = (b(C_k) - f(C_k))\cdot R(C_k)$. \Comment{Requires $f(C_k)$ to be \\ already computed}
	        \State $\yv(\tilde{v})  \gets \yv(\tilde{v}) + \Delta_k$ for all $\tilde{v} \in \wtd{C}_k$.
	        \State Update values of $f(C_j)$ for all $j\in\{k+1,\ldots, l\}$.
	    \EndFor
	    \For{$i \in V$}
	        \State Let $\tilde{v}(i)$ be the vertex in $\wtd{G}$ that $i$ was contracted to.
	         \State $\xv^t(i) \gets \xv^{t-1}(i) + \yv(\tilde{v}(i))$.
	    \EndFor
        \State Recompute $f(C)$ for all fundamental cuts $C$ of $T$.
	\EndFor
	\State Let $\fv^{\ceil{K/l}}$ be the tree-defined flow with respect to $\mathbf{x}^{\ceil{K/l}}$ and $T$.
	\State \Return{$\mathbf{x}^{\ceil{K/l}}$, $\fv^{\ceil{K/l}}$}
\end{algorithmic}
\end{algorithm}

\begin{theorem}
    The running time of \ALGNAME\ with batching is $O(m^{1.5}\log n \log\log n \log\frac{1}{\epsilon})$. This is achieved by choosing $l = \sqrt{m}$. 
\end{theorem}
\begin{proof}



At the beginning of the algorithm, we compute the values of $b(C)$ ($O(n)$ time via a dynamic program) and $R(C)$ ($O(m \log n)$ time via link-cut trees), just as in the proof of \Cref{lem:runtime}.

Consider a batch of $l$ updates. Note that the contracted tree $\wtd{T}$ has at most $l+1$ vertices. After contracting, we need to perform $l$ updates. This involves, for each $k \in \{1,2,\ldots, l\}$:
\begin{itemize}
    \item Computing $\Delta_k := (b(C_k) - f(C_k))\cdot R(C_k)$,
    \begin{itemize}
        \item This takes $O(1)$ time assuming $f(C_k)$ has already been computed. (Recall that the values $b(C_k)$ and $R(C_k)$ are computed at the very beginning of the algorithm.)
    \end{itemize}
    \item Adding $\Delta_k$ to $\yv(\tilde{v})$ for every $\tilde{v} \in \wtd{C}_k$.
    \begin{itemize}
        \item This takes $O(l)$ time, because the contracted tree has size $O(l)$. 
    \end{itemize}
    \item Updating the values $f(C_{k+1}), f(C_{k+2}), \ldots, f(C_{l})$ so they can be used in the later iterations of the inner loop. 
    \begin{itemize}
        \item If each $f(C_j)$ can be updated in $O(1)$ time, this takes $O(l)$ time. 
        \item To update each $f(C_j)$ in $O(1)$ time, we can precompute at the beginning of the block the $H(C_i, C_j)$ table for $i,j \in \{1,2,\ldots, l\}$, just as in the proof of \Cref{lem:runtime}. The difference now is that we only need to compute the table for the cuts that will be updated in the block. There are $l$ such cuts, so the total time to compute the table is $O(l^2)$, again using Karger's method.

    \end{itemize}
\end{itemize}

At the end of each block, we propagate the updates we made on the contracted graph back to the original graph. This involves
\begin{itemize}
    \item Determining the new potential of each node in $G$.
    \begin{itemize}
        \item This takes $O(n)$ time by iterating over the nodes of $G$.
    \end{itemize}
    \item Determining the value of $f(C)$ for each fundamental cut determined by $T$. 
    (By convention, assume the edges of $T$ are directed toward the root, and that the fundamental cuts we consider are the vertex sets of the subtrees of $T$.)
    \begin{itemize}
        \item This can be done in $O(m)$ time using a dynamic program that works from the leaves to the root. First, we compute $f(C)$ at each leaf of the tree. Next, suppose we have are at a non-leaf node $v$, and let $C$ be the set of vertices in the subtree rooted at $v$. Suppose we have already computed $f(D_1), f(D_2), \ldots, f(D_k)$, where $D_1, \ldots, D_k$ are the proper subtrees of $v$. Then we can compute $f(C)$ as follows:
        $$f(C) = \sum_{i=1}^k f(D_i) + \sum_{w: vw \in E} \frac{p(v) - p(w)}{r(v, w)}.$$
        This sum correctly counts the flow leaving $C$. This is because any edge leaving $C$ is counted once. On the other hand, if an edge is between $D_i$ and $D_j$, then it is counted once in the $f(D_i)$ term, and once with the \emph{opposite sign} in the $f(D_j)$ term, so it zeros out. Similarly, if an edge is between $D_i$ and $v$, it also zeros out. 
        
        The running time of this dynamic program is $O(m)$, because the time taken at each node is proportional to its degree, and the sum of all the node degrees is equal to $2m$. 
    \end{itemize}
\end{itemize}

To summarize, there are $K$ iterations, divided into blocks of size $l$. In each block, we pay the following.
\begin{itemize}
    \item \textbf{Start of block:} $O(m)$ time to contract the tree, and $O(l^2)$ time to compute the $H(C, C')$ table for the cuts that will be updated in the block.
    \item \textbf{During the block:} $O(l)$ time per iteration. Since each block consists of $l$ iterations, this is $O(l^2)$ in total.
    \item \textbf{End of block:} $O(m)$ time to propagate the changes from the contracted tree to the original tree.
\end{itemize}
Hence, each block takes $O(m + l^2)$ time. Multiplying by the number of blocks, which is $K/l$, this gives a running time of $O(K(\frac{m}{l} + l))$. Choosing $l = \sqrt{m}$ to minimize this quantity, we get $O(K\sqrt{m})$. 

The final running time is therefore $O(K\sqrt{m})$ plus the preprocessing time. 
Preprocessing consists of finding a a low-stretch spanning tree ($O(m\log n \log\log n$)), plus computing the values of $R(C)$ ($O(m\log n)$), and $f(C)$ ($O(n)$).

Thus, the preprocessing time is dominated by the time it takes to run the iterations. So, the total running time is now: $O(K\sqrt{m}) = O(m^{1.5}\log n \log\log n \log \frac{1}{\epsilon})$.

\end{proof}

\subsection{A Still Faster Algorithm via Batching, Sparsification, and Recursion}
\label{sec:recurse}
We now show that we can further speed up the algorithm using sparsification and recursion. The goal is to show that we can we can obtain a running time of the form $O(A^{\frac{1}{\delta}}m^{1+\delta}(\log n)^{\frac{B}{\delta}}(\log \frac{1}{\epsilon})^{\frac{1}{\delta}})$ for any $\delta > 0$, where $A$ and $B$ are numerical constants. 

Consider batching the iterations of the algorithm as follows. Pick a positive integer $d$, and repeat $K$ times:
\begin{itemize}
    \item Sample the next $d$ updates to be performed by the algorithm. These correspond to $d$ edges of the spanning tree $T$.
    \item Let $V_0, V_1, \ldots, V_d$ be the vertex sets that $T$ is partitioned into by the $d$ tree edges.
    \item Add $\Delta(i)$ to every vertex in $V_i$. We will choose the values $\Delta(0), \Delta(1), \ldots, \Delta(d)$ to greedily maximize the increase in the dual bound. 
\end{itemize}
Note that our original algorithm corresponds to the case when $d = 1$.  The lemma below quantifies the increase of the dual objective after one step of the above update.
\begin{lemma}
\label{lem:dual_incr}
Let $(V_0, \ldots, V_d)$ be a partition of $V$. 
Let $\xv \in \R^V$ be a vector of potentials, and let $\Delta = (\Delta(0), \ldots, \Delta(d))$ be any vector in $\R^{d+1}$. Let $\tilde{\xv}$ be obtained from $\xv$ by adding $\Delta(i)$ to the potential of every node in $V_i$. Then, the increase in the dual bound is given by the formula
$$\mathcal{B}(\wtd{\xv}) - \mathcal{B}(\xv) = {\bv}_H^T\Delta - \frac12\Delta^T{\Lm}_H\Delta,$$
where
\begin{itemize}
    \item $H$ is the contracted graph with vertices $V_0, V_1, \ldots, V_d$ and resistances $r(V_k, V_l) = \left(\sum_{ij \in \delta(V_k, V_l)} \frac{1}{r(i,j)}\right)^{-1}$,
    \item ${\Lm}_H$ is the Laplacian matrix of ${H}$, and
    \item ${b}_H(k) = b(V_k) - f(V_k)$ for $k = 0,1,\ldots, d$.
\end{itemize}
In particular, the choice of $\Delta$ that maximizes $\sB(\wtd{\xv}) - \sB(\xv)$ is given by the solution to ${\Lm}_H\Delta = \bv_H$.
\end{lemma}

\begin{proof}
We write the increase in the dual potential bound. Recall that $f(V_k)$ is the amount of flow leaving $V_k$ in the flow $f(i,j) = \frac{x(i)-x(j)}{r(i,j)}$. We let $f(V_k, V_l)$ be the amount of flow going from $V_k$ to $V_l$.
\begin{align*}
    &2\left(\sB(\wtd{\xv}) - \sB(\xv)\right) \\
    =\; &(2\bv^T\wtd{\xv} - \wtd{\xv}^T\Lm\wtd{\xv}) - (2\bv^T\xv- \xv^T\Lm\xv) \\
    =\; &2\sum_k b(V_k)\Delta(k) + \sum_{(i,j) \in \vec{E}} \frac{1}{r(i,j)}\left[(x(i)-x(j))^2 - (\tilde{x}(i) - \tilde{x}(j))^2\right] \\
    =\; & 2\sum_k b(V_k)\Delta(k) + \sum_{(i,j) \in \vec{E}} \frac{1}{r(i,j)}\left[(x(i)-x(j) + \tilde{x}(i) - \tilde{x}(j))(x(i)-x(j) - \tilde{x}(i) + \tilde{x}(j))\right] \\
    =\; & 2\sum_k b(V_k)\Delta(k) + \sum_{k<l}\sum_{(i,j)\in \delta(V_k,V_l)} \frac{1}{r(i,j)}\left[(2x(i)-2x(j) +\Delta(k)  -\Delta(l))(\Delta(l) - \Delta(k))\right] \\
    =\; & 2\sum_k b(V_k)\Delta(k) + 2\sum_{k<l}(\Delta(l) - \Delta(k))\sum_{(i,j)\in \delta(V_k,V_l)} \frac{1}{r(i,j)}(x(i)-x(j)) \\ & \qquad\qquad -\sum_{k<l}(\Delta(k) - \Delta(l))^2\sum_{(i,j) \in \delta(V_k, V_l)} \frac{1}{r(i,j)}  \\
    =\; & 2\sum_k b(V_k)\Delta(k) + 2\sum_{k<l}(\Delta(l) - \Delta(k))f(V_k, V_l) -\Delta^T\Lm_{H}\Delta  \\
    =\; & 2\sum_k b(V_k)\Delta(k) - 2\sum_{k}\Delta(k) f(V_k) -\Delta^T\Lm_{H}\Delta  \\
    =\; & 2\sum_k (b(V_k)-f(V_k))\Delta(k) -\Delta^T\Lm_{H}\Delta \\
    =\; & 2{\bv}_H^T\Delta -\Delta^T\Lm_{H}\Delta
\end{align*}
Note that this is a concave function of $\Delta$, because $\Lm_{H}$ is positive semidefinite. Therefore, maximizing this expression is equivalent to setting its gradient to 0. Taking its gradient and setting to 0 yields $\Lm_{H}\Delta = \bv_H$, as claimed. 

\end{proof}

\begin{remark}
Another interpretation of the $\Delta$ that maximizes $\sB(\wtd{\xv}) - \sB(\xv)$ in the Lemma above is as follows: $(\Delta(0), \ldots, \Delta(d))$ are the values such that if one adds $\Delta(i)$ to the potential of every vertex in $V_i$, the resulting potential-induced flow satisfies the flow constraints $f(V_k) = b(V_k)$ for all $k = 0, \ldots, d$. 
\end{remark}



\subsection{The Sparsify and Recurse Algorithm}
Next we give the algorithm with sparsification and recursion in more detail. Observe that $d$ cut-toggling updates effectively break the spanning tree into $d+1$ components.  After contracting the components to get a graph $H$ with $d+1$ vertices, \Cref{lem:dual_incr} shows that solving the Laplacian system $\Lm_H \Delta = \bv_H$ gives the update that maximizes the increase in $\sB(\wtd{\xv}) - \sB(\xv)$  among all updates that increment the potential of all vertices in $V_i$ by the same amount.  In particular, the progress made by this update step is is at least as large as the progress made by the sequence of $d$ updates performed by the straightforward unbatched algorithm. 

A natural approach is to solve $\Lm_H \Delta = \bv$ recursively.
However, this by itself does not give an improved running time.
Instead, we will first spectrally sparsify $H$ to get a sparsified approximation $\wtd{\Lm}_H$ of $\Lm_H$, satisfying $(1-\gamma)\Lm_H \preceq \wtd{\Lm}_H \preceq (1+\gamma) \Lm_H$ for an appropriate constant $\gamma \in (0,1)$.
Such a matrix $\wtd{\Lm}_H$ is known as a $\gamma$-spectral sparsifier of $\Lm_H$.
We then call the algorithm recursively {on $H$} to solve $\wtd{\Lm}_H \wtd{\Delta} = \bv_H$.
 This approach 
 is  akin to
recursive sparsification based Laplacian solvers~\cite{KMP12,JS21,CKMPPRX14},
but whose progress is bounded via the effect of single cut toggles. 
A main task of the analysis is to bound the error incurred by solving the sparsified system instead of the exact one.

For the spectral sparsification, we need to use an algorithm that does not require calling Laplacian solvers as a subroutine. For examples of such algorithms, see~\cite{kyng17sparsify,Koutis16sparsify,ST11:journal,kapralov12sparsify,PS13}. We will use the sparsifier given in Theorem 1.2 of~\cite{kyng17sparsify}, which returns a $\gamma$-spectral sparsifier with $O(n\log^2 n \log\log n / \gamma^2)$ edges in $O(m \log^2 n \log\log n / \gamma^2)$ time, with probability at least $1 - \frac{1}{\poly(n)}$. Here, $m,n$ are the number of edges and vertices, respectively, in the graph before sparsifying, and the $1- \frac{1}{\poly(n)}$ probability means that one can make the probability $1 - \frac{1}{n^k}$ for any constant $k$, with the $k$ appearing as a multiplicative factor in the big-$O$ expressions. We will take $k=2$ when calling the sparsifier in our algorithm, so it succeeds with probability at least $1 -\frac{1}{n^2}$.

\begin{algorithm}[!ht]
\begin{algorithmic}[1]
\State If $\abs{V} \leq n_0$, solve $\Lm_G\xv = \bv$ using Gaussian elimination and \textbf{return} $\xv$.
\\\Comment{$\Lm_G$ is the Laplacian matrix of $G$.}
\State Compute a low-stretch spanning tree $T$ of $G$. 
\State Initialize $\xv^0 = 0$.
\For{$t=0$ to $K$}
\State Generate the next $d$ updates. These correspond to $d$ tree edges $e_1^t, \ldots, e_d^t \in T$.
\State Let $V^t_0, \ldots, V^t_d$ be the vertex sets of the connected components of \\ $T - \{e_1^t, \ldots, e_d^t\}$. 
\State Contract $G$ to $H^t$ with $d+1$ vertices: Each $V_i^t$ is a vertex in $H^t$, and \\ resistances in $H^t$ are 
$$r_{H^t}(V_k^t, V_l^t) = \left(\sum_{ij \in \delta(V_k, V_l)} \frac{1}{r(i,j)}\right)^{-1}.$$
\State Compute $\wtd{\Lm}_{H^t}$, a $\gamma$-spectral sparsifier of $\Lm_{H^t}$, where $\gamma \in (0,1)$ will be \\ a parameter that we will determine later.
\State  Let $\bv^t_G = \bv - \Lm_G\xv^t$.
\State Let $\bv^t_{H^t} \in \R^{V(H^t)}$ be defined as follows:
    $b^t_{H^t}(V_i^t) = \sum_{u \in V_i^t} b^t_G(u) \quad \text{for all $i=0,1,\ldots,d$}.$
\State Call the algorithm recursively to solve the Laplacian system $\wtd{\Lm}_{H^t}\wtd{\Delta}^t = \bv^t_{H^t}$ for $\wtd{\Delta}^t$. This will return an approximate solution $\wtd\Delta^t_{\epsilon'}$ that satisfies $\norm{\wtd\Delta^t_{\epsilon'} - \wtd\Delta^t}_{\wtd{\Lm}_{H^t}}^2 \leq \epsilon' \norm{\wtd{\Delta}^t}^2_{\wtd{\Lm}_{H^t}}$. Here, $\epsilon'$ is the error parameter that we will input to the recursive call, to be determined later.
\State Update $\xv^t$ using $\tilde{\Delta}^t_{\epsilon'}$ to get the next iterate $\xv^{t+1}$. For every vertex $u \in V$, 
$$x^{t+1}(u) \gets x^t(u) + \tilde{\Delta}^t_{\epsilon'}(V^t_i),$$ where $V^t_i$ is the set in $V_0^t, \ldots, V_d^t$ such that $u \in V^t_i$. In other words, we update $\xv^t$ to $\xv^{t+1}$ by 
adding $\tilde{\Delta}^t_{\epsilon'}(V^t_i)$ to the potential of every vertex in $V^t_i$. 
\State \textbf{If} $\sB(\xv^{t+1}) \leq \sB(\xv^t)$, revert $\xv^{t+1} \gets \xv^t$. 
\EndFor
\State Return $\xv^K$ and the corresponding tree-defined flow $\fv^K$.
\caption{\ALGNAME\ with batching, sparsification, and recursion}
\label{alg:fastest}
\end{algorithmic}
\end{algorithm}
Pseudocode for \ALGNAME\ with batching, sparsification, and recursion is given in Algorithm \ref{alg:fastest}. The base case of the recursive algorithm is when $\abs{V} \leq n_0$, where $n_0$ is a constant that the algorithm can choose. For the base case, we simply use Gaussian elimination to solve $\Lm_G\xv =\bv$, which takes $O(1)$ time since $n_0$ is a constant. For every $t$, contracting $G$ down to $H^t$ and computing the new resistances takes $O(m)$ time; this is because each edge in $G$ contributes to the resistance between exactly one pair of nodes in $H^t$.

\subsection{Analysis of the Sparsify and Recurse Algorithm}
We now analyze Algorithm \ref{alg:fastest}. We first bound the convergence rate, then analyze the running time. 

\subsubsection{Error Analysis}
The lemma below bounds the expected rate of convergence of $\xv^t$ to $\xv^*$.
\begin{lemma}
\label{lem:err}
For all $t \geq 0$, we have $\E\norm{\xv^* - \xv^t}_{\Lm_G}^2 \leq \left(1 - \beta + \beta e^{-\frac{d}{\tau}}\right)^t \norm{\xv^*}_{\Lm_G}^2$. Here,
\begin{itemize}
    \item $\tau = O(m\log n \log\log n)$ is the stretch of the spanning tree,
    \item $d$ is the number of updates in each batch,
    \item $\beta = \left(1-\frac{1}{n_0^2}\right)\left(1-\left(4\epsilon'\cdot\frac{1+\gamma}{1-\gamma}\cdot\left(1+\frac{\gamma^2}{(1-\gamma)^2}\right) + \frac{2\gamma^2}{(1-\gamma)^2}\right)\right).$
\end{itemize}
In particular, if we choose $n_0 = 10$, $\gamma = \frac{1}{100}$, and $\epsilon' = \frac{1}{100}$, then $\beta \geq \frac{4}{5}$, so that
$$\E\norm{\xv^* - \xv^t}_{\Lm_G}^2 \leq \left(\frac{1}{5} + \frac{4}{5} e^{-\frac{d}{\tau}}\right)^t \norm{\xv^*}_{\Lm_G}^2.$$
\end{lemma}
\begin{proof}
Define the random variable $D^t := \sB(\xv^*) - \sB(\xv^t)$.
We will show in Lemma \ref{lem:gap_reduce} that for every possible realization $\xv^t$, we have 
$\E\left[D^{t+1} \mid \mathbf{x}^t\right] \leq \left(1 -\beta+\beta e^{-d/\tau}\right)\E\left[D^t \mid \mathbf{x}^t\right].$
This implies that $\E[D^{t+1}] \leq \left(1 -\beta+\beta e^{-d/\tau}\right)\E\left[D^t\right]$ unconditionally. 

It then follows that
$$\E\left[D^t\right] \leq \left(1 -\beta+\beta e^{-d/\tau}\right)^t\E\left[D^0\right] 
= \left(1 -\beta+\beta e^{-d/\tau}\right)^t\sB(\mathbf{x}^*).$$
Thus, 
$\sB(\mathbf{x}^*) - \E[\sB(\mathbf{x}^t)] \leq  \left(1 -\beta+\beta e^{-d/\tau}\right)^t\sB(\mathbf{x}^*).$
\end{proof}

As in the original analysis of \ALGNAME, we will study the duality gap and analyze its decrease at each step of the algorithm. Consider some iteration $t$ of the algorithm. Recall that $\xv^t$ is the iterate at the start of iteration $t$. For every possible sequence of $e_1^t, \ldots, e_d^t$ (the trees edges chosen in iteration $t$), define the following:
\begin{itemize}
    \item Let $\hat{\xv}^{t+1}$ be the vector obtained from $\xv^t$ by adding ${\Delta}^t(V^t_i)$ to every vertex in $V^t_i$, where  ${\Delta}^t := \Lm_{H^t}^\dag \bv^t_{H^t}$. 
    \item Let $\bar{\xv}^{t+1}$ be obtained from $\xv^t$ by applying the updates for the sequence of tree edges $e_1^t, \ldots, e_d^t$, \emph{one by one}. (i.e. Exactly as in the original, unbatched version of \ALGNAME\ described in \Cref{sec:alg}.)
\end{itemize}

\begin{lemma}
\label{lem:errs}
Fix any choice of $e^t_1, \ldots, e^t_d$, and assume that $\wtd{\Lm}_{H^t}$ is a $\gamma$-approximate sparsifier of $\Lm_{H^t}$. Then
$$\sB(\xv^{t+1}) - \sB(\xv^t) \geq  (1-\alpha)\left(\sB(\hat{\xv}^{t+1}) - \sB(\xv^t)\right)$$
where $\alpha = 4\epsilon'\cdot\frac{1+\gamma}{1-\gamma}\cdot\left(1+\frac{\gamma^2}{(1-\gamma)^2}\right) + \frac{2\gamma^2}{(1-\gamma)^2}$. 

If we further assume that $\gamma \in (0,\frac12)$, we can simplify to get
$$\sB(\xv^{t+1}) - \sB(\xv^t) \geq  (1-12\epsilon'-8\gamma^2-48\epsilon'\gamma^2)\left(\sB(\hat{\xv}^{t+1}) - \sB(\xv^t)\right).$$
\end{lemma}
\begin{proof}
To simplify notation, in this proof we will use 
\begin{itemize}
    \item $\bv_H$ to denote ${\bv^t_{H^t}}$, 
    \item $\tilde{\Delta}_{\epsilon'}$ to denote $\tilde{\Delta}_{\epsilon'}^t$,
    \item $\tilde{\Delta}$ to denote $\tilde{\Delta}^t$
    \item $\Delta$ to denote $\Delta^t$,
    \item $H$ to denote $H^t$,
\end{itemize}

Later in this proof, we will show that 
\begin{equation}
\label{eq:delta_approx}
\norm{\tilde{\Delta}_{\epsilon'} - \Delta}_{\Lm_H}^2 \leq \alpha \norm{\Delta}_{\Lm_H}^2,
\end{equation}
for a constant $\alpha$ that depends on $\epsilon'$ and $\gamma$. Assuming (\ref{eq:delta_approx}) holds, by the definition of the matrix norm it follows that
$$\left(\tilde{\Delta}_{\epsilon'} - \Delta\right)^T\Lm_H\left(\tilde{\Delta}_{\epsilon'} - \Delta\right) \leq \alpha \Delta^T\Lm_H\Delta.$$
Expanding the left-hand side and rearranging, we get
$$2\tilde{\Delta}_{\epsilon'}^T\Lm_H \Delta - \tilde{\Delta}_{\epsilon'}^T\Lm_H \tilde{\Delta}_{\epsilon'} \geq (1 - \alpha) \Delta^T\Lm_H\Delta.$$
Using $\Lm_H \Delta = \bv_H$, this becomes
$$2\tilde{\Delta}_{\epsilon'}^T\bv_H - \tilde{\Delta}_{\epsilon'}^T \Lm_H\tilde{\Delta}_{\epsilon'} \geq (1 - \alpha) \Delta^T\Lm_H\Delta.$$
Recall that $\xv^{t+1}$ is obtained from $\xv^t$ by adding $\tilde{\Delta}_{\epsilon'}(V_i^t)$ to every vertex in $V_i^t$. 
Using Lemma \ref{lem:dual_incr}
with $\Delta(i) = \wtd{\Delta}_{\epsilon'}(V_i^t)$, $\xv = \xv^t$,  $\wtd{\xv} = \xv^{t+1}$, $\wtd{\Lm} = \Lm_H$, and $\wtd{\bv} = \bv_H$ it follows that the left-hand side is equal to $\sB(\xv^{t+1}) - \sB(\xv^t)$. 
On the other hand, $\sB(\hat{\xv}^{t+1}) - \sB(\xv^t) = 2\bv_H^T\Delta - \Delta^T\Lm_H\Delta = \Delta^T\Lm_H\Delta$. (Since $\Lm_H \Delta = \bv_H$.)
Thus, the right-hand side is equal to $(1-\alpha)\left(\sB(\hat{\xv}^{t+1}) - \sB(\xv^{t})\right)$. Thus we have
$$\sB(\xv^{t+1}) - \sB(\xv^t) \geq  (1 - \alpha)\left(\sB(\hat{\xv}^{t+1}) - \sB(\xv^t)\right),$$
as claimed. 

It remains to prove (\ref{eq:delta_approx}). To prove (\ref{eq:delta_approx}), note that we have
\begin{enumerate}
    \item $\norm{\tilde{\Delta}_{\epsilon'} - \tilde{\Delta}}_{\wtd{\Lm}_H}^2 \leq \epsilon' \norm{\tilde{\Delta}}_{\wtd{\Lm}_H}^2$ (This is the error from the recursive solve).
    \item $\norm{\tilde{\Delta} - \Delta}^2_{\Lm_H} \leq h(\gamma) \norm{\Delta}_{\Lm_H}^2$ (Follows by part 2 of Proposition \ref{prop:spec_approx} in the Appendix. This is the error from sparsification). 
\end{enumerate}
The first inequality, together with $(1-\gamma)\Lm_H \preceq \wtd{\Lm}_H \preceq (1+\gamma)\Lm_H$ and part 1 of Proposition \ref{prop:spec_approx} in the Appendix, implies that 
$$\norm{\tilde{\Delta}_{\epsilon'} - \tilde{\Delta}}_{\Lm_H}^2 \leq
\frac{1}{1-\gamma}
\norm{\tilde{\Delta}_{\epsilon'} - \tilde{\Delta}}_{\wtd{\Lm}_H}^2 \leq
\frac{\epsilon'}{1-\gamma} \norm{\tilde{\Delta}}_{\wtd{\Lm}_H}^2\leq \epsilon' \cdot \frac{1+\gamma}{1-\gamma} \cdot \norm{\tilde{\Delta}}^2_{\Lm_H}.$$

Now, using the inequality $\norm{a+b}^2 \leq 2\norm{a}^2 + 2\norm{b}^2$ (which holds for any norm), we note that
$$\norm{\tilde{\Delta}}^2_{\Lm_H} \leq 2\norm{\Delta}_{\Lm_H}^2 + 2\norm{\tilde{\Delta} -\Delta}_{\Lm_H}^2 \leq 2\norm{\Delta}_{\Lm_H}^2 +2h(\gamma)\norm{\Delta}_{\Lm_H}^2.$$
Hence,
$$\norm{\tilde{\Delta}_{\epsilon'} - \tilde{\Delta}}_{\Lm_H}^2 \leq 2\epsilon' \cdot \frac{1+\gamma}{1-\gamma} \cdot (1+h(\gamma)) \norm{\Delta}_{\Lm_H}^2.$$
Again using $\norm{a+b}^2 \leq 2\norm{a}^2 + 2\norm{b}^2$, we have
\begin{align*}
    \norm{\tilde{\Delta}_{\epsilon'} - {\Delta}}_{\Lm_H}^2
    &\leq 2\left(\norm{\tilde{\Delta}_{\epsilon'} - \tilde{\Delta}}_{\Lm_H}^2 + \norm{\tilde{\Delta} - {\Delta}}_{\Lm_H}^2\right) \\
    &\leq 2\left(2\epsilon' \cdot \frac{1+\gamma}{1-\gamma} \cdot (1+h(\gamma)) \norm{\Delta}_{\Lm_H}^2 + h(\gamma)\norm{\Delta}_{\Lm_H}^2\right) \\
    &= \left(4\epsilon' \cdot \frac{1+\gamma}{1-\gamma} \cdot (1+h(\gamma)) + 2h(\gamma) \right)\norm{\Delta}_{\Lm_H}^2.
\end{align*}
Therefore, (\ref{eq:delta_approx}) holds with $\alpha =4\epsilon' \cdot \frac{1+\gamma}{1-\gamma} \cdot (1+h(\gamma)) + 2h(\gamma)$.
\end{proof}
\begin{lemma}
\label{lem:gap_reduce}
For any vector $\xv^t$, we have
$$\sB(\xv^*) - \E[\sB({\xv}^{t+1})] \leq \left(1 -\beta+\beta e^{-d/\tau}\right) \left(\sB(\xv^*) - \sB(\xv^t)\right),$$
where $\beta = (1-\frac{1}{n_0^2})(1-\alpha)$. 

Here, the expectation is taken over the random choices of $e_1^t, \ldots, e^t_d$, and also over the randomness of the sparsification step. (Recall that we use the sparsifier given in Theorem 1.2 of~\cite{kyng17sparsify}, which successfully returns a $\gamma$-approximate sparsifier with probability at least $1 - \frac{1}{\abs{V(H^t)}^2}$.)

\end{lemma}
\begin{proof}
By Lemma \ref{lem:gap_decreases}, we know that
$$\sB(\xv^*) - \E[\sB(\bar{\xv}^{t+1})] \leq \left(1 - \frac1\tau\right)^d\left(\sB(\xv^*) - \sB(\xv^t)\right) \leq e^{-\frac{d}{\tau}}\left(\sB(\xv^*) - \sB(\xv^t)\right).$$
Rearranging, this is equivalent to
$$\E[\sB(\bar{\xv}^{t+1})] - \sB(\xv^t) \geq \left(1 - e^{-\frac{d}{\tau}} \right)\left(\sB(\xv^*) - \sB(\xv^t)\right).$$
Observe that for every realization of $e^t_1, \ldots, e^t_d$, we have $\sB(\hat{\xv}^{t+1}) \geq \sB(\bar{\xv}^{t+1})$. This is because $\hat{\xv}^{t+1} - \xv^t = \Delta^t$, where $\Delta^t$ by definition is the vector that \emph{maximizes} the increase $\sB(\hat{\xv}^{t+1}) - \sB(\xv^t)$ while subject to being incremented by the same amount on each of the components $V^t_0, \ldots, V^t_d$. On the other hand, the vector $\bar{\xv}^{t+1} - \xv^t$ is also  incremented by the same amount on each of the components $V^t_0, \ldots, V^t_d$ by the way our original algorithm works.

Since $\sB(\hat{\xv}^{t+1}) \geq \sB(\bar{\xv}^{t+1})$ holds for every realization of $e^t_1, \ldots, e^t_d$, it follows that $\E[\sB(\hat{\xv}^{t+1})] \geq \E[\sB(\bar{\xv}^{t+1})]$, where the expectation is taken over the random choices of $e^t_1, \ldots, e^t_d$ made by the algorithm. Hence,
\begin{equation}
\label{eq:opt_incr}
\E[\sB(\hat{\xv}^{t+1})] - \sB(\xv^t) \geq \left(1 - e^{-\frac{d}{\tau}}\right) \left(\sB(\xv^*) - \sB(\xv^t)\right).
\end{equation}
To conclude, we will use Lemma \ref{lem:errs} to translate the above inequality (which is in terms of $\hat{\xv}^{t+1}$), to an inequality in terms of $\xv^{t+1}$. We have
\begin{itemize}
    \item With probability $\geq 1 - \frac{1}{n_0^2}$, the sparsifier is successful and by Lemma \ref{lem:errs}, $$\E[\sB(\xv^{t+1})] - \sB(\xv^t) \geq (1-\alpha)\left(\E[\sB(\hat{\xv}^{t+1})] - \sB(\xv^t) \right).$$
    \item With probability $\leq \frac{1}{n_0^2}$, the sparsifier is unsuccessful and $\E[\sB(\xv^{t+1})] - \sB(\xv^t) \geq 0$. This is because in the algorithm, we evaluate $\sB(\xv^{t+1})$ and only update $\xv^t$ to $\xv^{t+1}$ if $\sB(\xv^{t+1}) \geq \sB(\xv^t)$. Otherwise, we make $\xv^{t+1} = \xv^t$. 
\end{itemize}
Note that the above expectations are with respect to the random choices of $e^t_1, \ldots, e^t_d$, conditioned on the sparsifier being successful/unsuccessful. Now, taking another expectation with respect to the randomness of the sparsifier, we get
\begin{align*}
\E[\sB(\xv^{t+1})] - \sB(\xv^t) &\geq \left(1-\frac{1}{n_0^2}\right)(1-\alpha)\left(\E[\sB(\hat{\xv}^{t+1})] - \sB(\xv^t) \right) \\
&\geq \left(1-\frac{1}{n_0^2}\right)(1-\alpha)\left(1 - e^{-\frac{d}{\tau}}\right) \left(\sB(\xv^*) - \sB(\xv^t)\right) &\text{(by (\ref{eq:opt_incr}))}
\end{align*}
Rearranging the above inequality gives
$$\sB(\xv^*) - \E[\sB({\xv}^{t+1})] \leq \left(1 - \left(1-\frac{1}{n_0^2}\right)(1-\alpha)\left(1 - e^{-\frac{d}{\tau}}\right)\right) \left(\sB(\xv^*) - \sB(\xv^t)\right),$$
as claimed.
\end{proof}

\subsubsection{Running Time Analysis}
The following theorem bounds the running time of the algorithm.
\fastest*


\begin{proof}
By \Cref{lem:err}, it suffices to run the algorithm for $K$ iterations, for any $K \geq \frac{\ln \epsilon}{\ln\left(\frac{1}{5} + \frac{4}{5}e^{-d/\tau}\right)}$.

Using the inequalities $e^{-x} \leq 1 - \frac{x}{2}$ and $\ln(1-x) \leq -\frac{x}{2}$ which hold for $x\in(0,1)$, we see that it suffices to choose $K = \frac{5\tau}{d} \ln(1/\epsilon)$. 

Recall $\tau \leq c_3m\log^2n$, for some constant $c_3$. Thus, we choose 
$d = c_3\bar{m}(\log^2n)\cdot{m}^{-\delta}.$ Here, $\bar{m}$ is the number of edges in the \emph{current iteration}, while ${m}$ is the number of edges in the \emph{topmost} iteration (i.e. in the original graph $G$).
With this choice of $d$, we have $K = 5c_3m^{\delta}\ln(1/\epsilon).$ (Note that $K$ is the same at every level of the recursion tree.)

The work involved in each call to the algorithm consists of 
\begin{itemize}
    \item Computing $T$,
    \item Doing $K$ times
    \begin{itemize}
    \item Contracting and sparsifying to a graph with $d$ vertices and $a_1 d \log^cn / \gamma^2$ edges, for some constant $a_1$.
    \item Doing a recursive call.
    \end{itemize}
\end{itemize}
If $m$ is the number of edges in the graph at one level of the recursion tree, then at the next level, the number of edges is
$$a_1 d \log^2n\log\log n / \gamma^2 \leq a_1c_3\bar{m}(\log^2n)\cdot{m}^{-\delta}\log^3n / \gamma^2 = a_1c_3\bar{m}\cdot{m}^{-\delta}(\log n)^{5}/\gamma^2.$$
Therefore, the number of edges in a graph at level $l$ of the recursion tree is
$$\textrm{edges}_l \leq {m}^{1-l\delta}(a_1c_3)^l(\log n)^{5l}/\gamma^{2l}.$$
The total work required by a node at level $l$ of the recursion tree is dominated by the sparsifier, which takes time
$$\textrm{work}_l = K\cdot a_2\cdot \textrm{edges}_l \cdot \log^{2}n\log\log n/ \gamma^2 \leq a_2\cdot \textrm{edges}_l \cdot \log^{3}n/ \gamma^2$$
for some constant $a_2$.

Finally, the total number of recursion tree nodes at level $l$ is equal to $K^l$. This implies that the total work required by all the nodes at level $l$ of the recursion tree is equal to 
\begin{align*}
\textrm{total work}_l &= \textrm{work}_l \cdot K^l \\
&\leq K\cdot a_2\cdot \textrm{edges}_l \cdot \log^{3}n\cdot \frac{1}{\gamma^2} \cdot K^l \\
&\leq K\cdot a_2\cdot \left({m}^{1-l\delta}(a_1c_3)^l(\log n)^{5l}/\gamma^{2l}\right) \cdot \log^{3}n\cdot \frac{1}{\gamma^2} \cdot K^l\\
&= 5^{l+1}c_3^{l+1}{m}^{1+\delta}(\ln 1/\epsilon)^{l+1}a_2(a_1c_3)^l(\log n)^{(5l+3)} \cdot \frac{1}{\gamma^{2l+2}} \\
&\leq A^l {m}^{1+\delta}(\log n)^{8l}(\ln 1/\epsilon)^{l+1},
\end{align*}
for some numerical constant $A$.

To conclude, we note that the total work summed across all the levels is at most a constant factor times the total work at the maximum level, which is $l = \frac{1}{\delta}$. 

\end{proof}

\begin{remark}
By setting $\delta = \sqrt{\frac{8\log\log n + \log\log\frac{1}{\epsilon}}{\log m}}$ (which is the choice of $\delta$ that minimizes $m^{1+\delta}(\log n)^{\frac{8}{\delta}}(\log \frac{1}{\epsilon})^{\frac{1}{\delta}}$), the running time in \Cref{thm:time} becomes $O\left(m \exp\left(2\sqrt{\log m(8\log\log n + \log\log \frac{1}{\epsilon}) }\right)\right)$, which is $O(m^{1+o(1)})$ if $\epsilon$ is a constant.
\end{remark}

\section{Conclusion} \label{sec:conclusion}
We propose a cut-based combinatorial algorithm to solve Laplacian systems approximately. This algorithm is dual to the cycle-based algorithm by Kelner et al. \cite{KOSZ13}. We show that our algorithm converges in a near-linear number of iterations. 


To achieve a near-linear running time, we would further need each iteration to run in polylogarithmic time. We give evidence against this, by presenting a reduction from the OMv conjecture. This is in contrast to the algorithm in \cite{KOSZ13}, which uses a data structure such that each iteration of the algorithm runs in $O(\log n)$ time.   In order to obtain a better running time, one would need to  show it is possible take advantage of the particular structure of updates in the algorithm to implement the data structure.  Note that our reduction crucially needs that a very specific spanning tree (albeit with very small stretch) is chosen. Is it possible for the algorithm to choose a different small-stretch spanning tree that is amendable to a polylogarithmic time implementation?

\section{Acknowledgements}
Monika Henzinger was supported by funding from the European Research Council (ERC) under the European Union's Horizon 2020 research and innovation programme (Grant agreement No. 101019564 ``The Design of Modern Fully Dynamic Data Structures (MoDynStruct)'' and from the Austrian Science Fund (FWF) project ``Fast Algorithms for a Reactive Network Layer (ReactNet)'', P~33775-N, with additional funding from the \textit{netidee SCIENCE Stiftung}, 2020--2024. \phantom{} {\includegraphics[width=60pt]{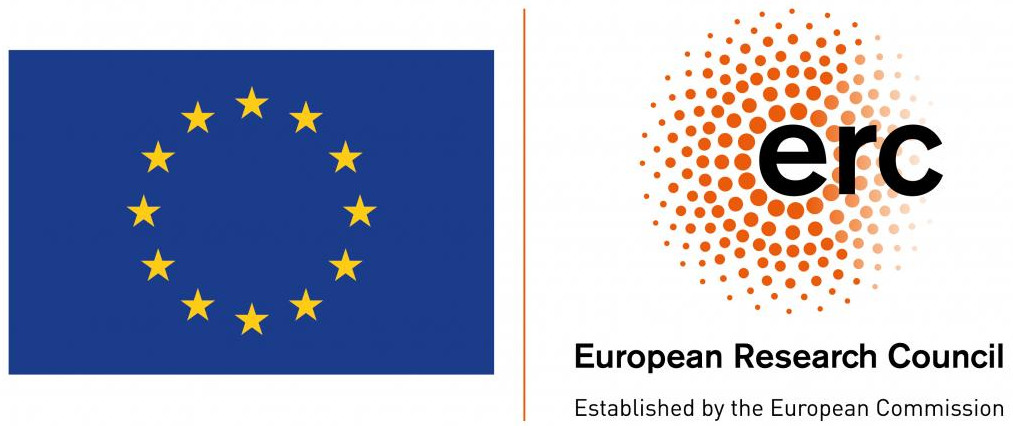}
}

Billy Jin was Supported in part by NSERC fellowship PGSD3-532673-2019 and NSF grant CCF-2007009.

Richard Peng was supported in part by an NSERC Discovery Grant and NSF grant CCF-1846218.

David P. Williamson was supported in part by NSF grant CCF-2007009.

\bibliographystyle{alpha} 
\bibliography{references}



    
\appendix

\section{Omitted Proofs from Section \ref{sec:alg}}

\enerincr*

\begin{proof}
The way we update $\mathbf{x}$ is by adding a constant $\Delta$ to the potentials of every vertex in $C$, where 
$$
\Delta = (b(C) - f(C))\cdot R(C)
$$
Recall that $f(C)$ is the net amount of flow going out of $C$ in the flow induced by $\mathbf{x}$. That is,
$$f(C) = \sum_{\substack{ij \in E \\ i \in C,\, j \not\in C}} \frac{x(i) - x(j)}{r(i, j)}$$
Note that the new potentials $\mathbf{x}'$ can be expressed as $\mathbf{x}' = \mathbf{x} + \Delta \one_C$. 
We have
\begin{align*}
    2(\sB(\mathbf{x}') - \sB(\mathbf{x}))
    &= 2\mathbf{b}\trans \mathbf{x}' - (\mathbf{x}')\trans \mathbf{L}\mathbf{x}' - (2\mathbf{b}\trans \mathbf{x} - \mathbf{x}\trans \mathbf{L}\mathbf{x}) \\
    &= 2\mathbf{b}\trans(\mathbf{x} + \Delta\cdot\one_C) - 2\mathbf{b}\trans \mathbf{x} - (\mathbf{x}')\trans \mathbf{L}\mathbf{x}' + \mathbf{x}\trans \mathbf{L}\mathbf{x} \\
    &= 2\Delta\cdot \mathbf{b}\trans \one_C - \sum_{ij \in E}\frac{1}{r(i, j)}\left[(x'(i) - x'(j))^2 - (x(i) - x(j))^2 \right] \\
    &= 2\Delta\cdot \mathbf{b}\trans \one_C - \sum_{ij \in \delta(C)}\frac{1}{r(i, j)}\left[(x'(i) - x'(j))^2 - (x(i) - x(j))^2 \right] \\
    &= 2\Delta\cdot \mathbf{b}\trans \one_C - \sum_{\substack{i \in C, \, j \not\in C \\ ij \in \delta(C)}}\frac{1}{r(i, j)}\left[(x(i)+\Delta - x(j))^2 - (x(i) - x(j))^2 \right] \\
    &= 2\Delta \cdot \mathbf{b}\trans \one_C - \sum_{\substack{i \in C, \, j \not\in C \\ ij \in \delta(C)}}\frac{1}{r(i,j)}\left[2\Delta\cdot(x(i)-x(j)) +\Delta^2\right] \\
    &= 2\Delta \cdot \mathbf{b}\trans \one_C - 2\Delta\cdot f(C) - \Delta^2 \sum_{(i,j) \in \delta(C)} \frac{1}{r(i,j)} \\
    &= 2\Delta \cdot \mathbf{b}\trans \one_C - 2\Delta\cdot f(C) - \Delta^2\cdot R(C)^{-1} \\
    &= 2\Delta \cdot b(C) - 2\Delta\cdot f(C) - \Delta^2 R(C)^{-1} \\
    &= 2\Delta ^2 R(C)^{-1} - \Delta^2 R(C)^{-1} \\
    &= \Delta^2/R(C).
\end{align*}
\end{proof}

\gapp*

\begin{proof}
By definition, we have
\begin{align*}
   2\gap(\mathbf{f}, \mathbf{x})
    &= \sum_{e \in E} r(e)f(e)^2 - (2\mathbf{b}\trans \mathbf{x} - \mathbf{x}\trans \mathbf{L}\mathbf{x}).
\end{align*}
Note that 
\begin{align*}
    \mathbf{b}\trans \mathbf{x} &= \sum_{i \in V}b(i)x(i) \\
    &= \sum_{i \in V} x(i) \left(\sum_{j: (i,j) \in \vec{E}} f(i,j) - \sum_{j:(j,i) \in \vec{E}} f(j,i)\right) \\
    &= \sum_{(i,j) \in \vec{E}}f(i,j)(x(i) - x(j))
    \end{align*}
    and 
    \begin{align*}
    \mathbf{x}\trans \mathbf{L}\mathbf{x} &= 
    \sum_{(i,j) \in \vec{E}} \frac{(x(i) - x(j))^2}{r(i,j)}.
    \end{align*}
    Plugging these into our expression for $\gap(\mathbf{f}, \mathbf{x})$, we obtain 
    \begin{align*}
        2\gap(\mathbf{f}, \mathbf{x})
        &= \sum_{(i,j) \in \vec{E}}\left[r(i,j)f(i,j)^2 - 2f(i,j)(x(i) - x(j)) + \frac{(x(i) - x(j))^2}{r(i,j)}\right] \\
        &= \sum_{(i,j) \in \vec{E}} r(i,j)\left(f(i,j) - \frac{x(i)-x(j)}{r(i,j)}\right)^2
    \end{align*}
    which is what we wanted to show.

\end{proof}

\gaplem*

\begin{proof}
Recall that $C(i,j)$, $\Delta(C(i,j))$ and $R(C(i,j))$ were defined as follows:
\begin{itemize}
    \item $C(i,j)$ is the set of vertices on the side of the fundamental cut of $T$ determined by $(i,j)$ containing $i$. In other words, $C(i,j)$ consists of the vertices in the component of $T - ij$ with $i \in C(i,j)$ and $j \not\in C(i,j)$.
    \item $R(C(i,j)) = \left(\sum_{ij \in \delta(C)} \frac{1}{r(i,j)}\right)^{-1}$.
    \item $\Delta(C(i,j)) = (b(C(i,j)) - f(C(i,j)))R(C(i,j))$, where
    \begin{itemize}
        \item $b(C(i,j)) = \mathbf{b}\trans \one_{C(i,j)}$, and 
        \item $f(C(i,j)) = \displaystyle\sum_{\substack{k \in C(i,j), \, l \not\in C(i,j) \\ kl \in E}} \frac{x(k) - x(l)}{r(k, l)}$
    \end{itemize}
\end{itemize}
We have 
\begin{align*}
    2\gap(\fv_{T, \xv}, \mathbf{x})
    &= \sum_{(i, j) \in E} r(i, j) \left(f_{T,\xv}(i, j) - \frac{x(i)-x(j)}{r(i, j)}\right)^2 \\
    &= \sum_{(i, j) \in T} r(i, j) \left(f_{T,\xv}(i, j) - \frac{x(i)-x(j)}{r(i, j)}\right)^2 \\
    &= \sum_{(i, j) \in T} r(i, j) \left[\left(b(C(i, j)) - \sum_{\substack{ k \in C(i, j), l \not\in C(i, j)\\kl \in E - ij }} \frac{x(k) - x(l)}{r(k, l)}\right) - \frac{x(i) - x(j)}{r(i, j)} \right]^2 \\
    &= \sum_{(i, j) \in T} r(i, j) \left[b(C(i,j)) - \sum_{\substack{k \in C(i, j), l \not\in C(i, j)\\kl \in E}} \frac{x(k) - x(l)}{r(k, l)}\right]^2 \\
    &= \sum_{(i, j) \in T} r(i, j) \left[b(C(i,j)) - f(C(i, j))\right]^2 \\
    &= \sum_{(i, j) \in T} r(i, j)\cdot \frac{\Delta(C(i, j))^2}{R(C(i, j))^2}
\end{align*}
\end{proof}

\gapdecr*

\begin{proof}
We know from the discussion above that
$$\E[\sB(\mathbf{x}^{t+1})] - \sB(\mathbf{x}^t) = \frac{1}{\tau}\gap(\fv_{T, \xv}, \mathbf{x}^t),$$
where $\fv_{T, \xv}$ is the tree-defined flow associated with potentials $\mathbf{x}^t$. Since $\gap(\fv_{T, \xv}, \mathbf{x}^t) \geq \sB(\mathbf{x}^*) - \sB(\mathbf{x}^t)$, we get
$$\E[\sB(\mathbf{x}^{t+1})] - \sB(\mathbf{x}^t) \geq \frac{1}{\tau}\left(\sB(\mathbf{x}^*) - \sB(\mathbf{x}^t)\right).$$
Rearranging gives
$$\sB(\mathbf{x}^*) - \E[\sB(\mathbf{x}^{t+1})] \leq \left(1 - \frac1\tau\right)\left(\sB(\mathbf{x}^*) - \sB(\mathbf{x}^t)\right),$$
as desired. 
\end{proof}

\finalgap*

\begin{proof}
Define the random variable $D_t := \sB(\mathbf{x}^*) - \sB(\mathbf{x}^t)$. By Lemma \ref{lem:gap_decreases}, we know that
$$\E\left[D^{t+1} \mid \mathbf{x}^t\right] \leq \left(1 - \frac{1}{\tau}\right)\E\left[D^t \mid \mathbf{x}^t\right]$$
for all possible vectors of potentials $\mathbf{x}^t$. This implies that $\E\left[D^{t+1}\right] \leq \left(1 - \frac{1}{\tau}\right)\E\left[D^t \right]$ unconditionally. 

By induction on $t$, it then follows that
$$\E\left[D^K\right] \leq \left(1 - \frac1\tau\right)^K\E\left[D^0\right] =\left(1 - \frac1\tau\right)^K\left(\sB(\mathbf{x}^*) - \sB(\mathbf{x}^0)\right) = \left(1 - \frac1\tau\right)^K\sB(\mathbf{x}^*).$$
Thus, 
$$\sB(\mathbf{x}^*) - \E[\sB(\mathbf{x}^K)] \leq  \left(1 - \frac1\tau\right)^K\sB(\mathbf{x}^*).$$

Using the inequality $1-x \leq e^{-x}$, we obtain
$$\sB(\mathbf{x}^*) - \E[\sB(\mathbf{x}^K)] \leq e^{-K/\tau} \sB(\mathbf{x}^*).$$
Hence, if $K \geq \tau\ln(\frac{1}{\epsilon})$, then we will have $\sB(\mathbf{x}^*) - \E[\sB(\mathbf{x}^K)] \leq \epsilon\cdot\sB(\mathbf{x}^*)$, as desired. 
\end{proof}

\dualkosz*
    \begin{proof}[Proof of \Cref{thm:dual_kosz}]
     By Corollary \ref{cor:final_gap}, after $K=\tau\ln(\frac{\tau}{\epsilon})$ iterations, the algorithm returns potentials  $\xv^K$ such that $\sB(\mathbf{x}^*) - \E[\sB(\mathbf{x}^K)] \leq \frac{\epsilon}{\tau} \cdot \sB(\mathbf{x}^*)$. Combining with Lemma \ref{lem:energy_to_potential}, we get that $\E \norm{\mathbf{x}^* - \mathbf{x}^K}_\mathbf{L}^2 \leq \frac{\epsilon}{\tau}\norm{\mathbf{x}^*}_\mathbf{L}^2$. Finally, \Cref{lem:rounding_error} gives $\E\left[\sE(\fv^K)\right] \leq (1+\epsilon)\sE(\fv^*)$.  
     
    \end{proof}

   \begin{restatable}{lemma}{stretch}
    \label{lem:stretch}
    We have $\tau = \st_T(G, \rv)$. 
    \end{restatable}

\begin{proof}
    We write out the definitions of $\tau$ and $\st_T(G, \rv)$:
    $$\tau = \sum_{(i, j) \in T} \frac{r(i, j)}{R(C(i, j))} = \sum_{(i, j) \in T} r(i, j)\sum_{(k,l) \in \delta(C(i, j))} \frac{1}{r(k, l)}$$
    and
    $$
    \st_T(G, \rv) = \sum_{(i, j) \in \vec{E}} \st_T((i, j), \rv)
    = \sum_{(i, j) \in \vec{E}} \frac{1}{r(i, j)}\sum_{(k, l) \in P(i,j)} r(k, l),
    $$
    where $P(i, j)$ is the unique path from $i$ to $j$ in $T$. 
    
    It turns out that the expressions for $\tau$ and $\st_T(G)$ are summing exactly the same terms, just in different ways. Indeed, we have
    \begin{align*}
        \tau &= \sum_{(i, j) \in T} \sum_{(k,l) \in \delta(C(i, j))}\frac{r(i, j)}{r(k, l)} \\
        &= \sum_{(k, l) \in \vec{E}} \sum_{(i, j) \in P(k, l)} \frac{r(i, j)}{r(k, l)} \\
        &= \st_T(G, \rv). 
    \end{align*}
    To switch the order of summation from the first line to the second line, we used the fact that for an edge $(k, l) \in \vec{E}$, we have $(k,l) \in \delta(C(i, j))$ if and only if $(i, j) \in P(k, l)$. This is because $T$ is a spanning tree. 
    
    \end{proof}

    By Corollary \ref{cor:final_gap}, we know that the potentials $\mathbf{x}^t$ found by the algorithm satisfy the property that $\sB(\mathbf{x}^t)$ converges to $\sB(\mathbf{x}^*)$ at a linear rate, in expectation. The following lemma shows that if $\mathbf{x}$ is a set of potentials such that $\sB(\mathbf{x})$ is close to $\sB(\mathbf{x}^*)$, then $\mathbf{x}$ is close to $\mathbf{x}^*$ as a vector (measured in the matrix norm defined by the Laplacian $\mathbf{L}$). 
    
    \begin{restatable}{lemma}{enertopot}
    \label{lem:energy_to_potential}
    Let $\mathbf{x}$ be any vector of potentials. Then 
    $\frac12\norm{\mathbf{x}^* - \mathbf{x}}_\mathbf{L}^2 = \sB(\mathbf{x}^*) - \sB(\mathbf{x}).$
    In particular, if $\sB(\mathbf{x}^*) - \sB(\mathbf{x}) \leq \epsilon \cdot \sB(\mathbf{x}^*)$, then $\norm{\mathbf{x}^* - \mathbf{x}}_\mathbf{L}^2 \leq  \epsilon \norm{\mathbf{x}^*}_\mathbf{L}^2.$
    \end{restatable}

\begin{proof}
    We have 
    \begin{align*}
        \norm{\mathbf{x}^* - \mathbf{x}}_L^2
        &= (\mathbf{x}^* - \mathbf{x})\trans \mathbf{L} (\mathbf{x}^* - \mathbf{x}) \\
        &= (\mathbf{x}^*)\trans \mathbf{L}\mathbf{x}^* - 2\mathbf{x}\trans \mathbf{L}\mathbf{x}^* + \mathbf{x}\trans \mathbf{L}\mathbf{x} \\
        &= 2\sB(\mathbf{x}^*) - 2\mathbf{x}\trans \mathbf{b} + \mathbf{x}\trans \mathbf{L}\mathbf{x} \\
        &= 2\sB(\mathbf{x}^*) - 2\sB(\mathbf{x}).
    \end{align*}
    In particular, if $\sB(\mathbf{x}^*) - \sB(\mathbf{x}) \leq \epsilon \cdot \sB(\mathbf{x}^*)$, then $\norm{\mathbf{x}^* - \mathbf{x}}_\mathbf{L}^2 \leq 2\epsilon \cdot \sB(\mathbf{x}^*) = \epsilon \norm{\mathbf{x}^*}_\mathbf{L}^2$.
    This is because
    $$2\sB(\mathbf{x}^*) = 2\mathbf{b}\trans \mathbf{x}^* - (\mathbf{x}^*)\trans \mathbf{L}\mathbf{x}^* = (\mathbf{x}^*)\trans \mathbf{L}\mathbf{x}^* =  \norm{\mathbf{x}^*}_\mathbf{L}^2.$$
    \end{proof}
    
    Next, we show that if $\sB(\xv)$ is sufficiently close to $\sB(\xv^*)$, then the associated tree-defined flow $\fv_{T, \xv}$ has energy sufficiently close to $\sE(\fv^*)$.  
    
    \begin{restatable}{lemma}{rounderror}
    \label{lem:rounding_error}
    For any distribution over $\xv$ such that $\E_{\xv}[\sB(\xv)] \geq (1-\frac{\epsilon}{\tau})\sB(\xv^*)$, we have $\E_{\xv}[\sE(\fv_{T, \xv})] \leq (1+\epsilon)\sE(\fv^*)$.
    \end{restatable}

\begin{proof}
For ease of notation, in this proof let $\fv = \fv_{T, \xv}$. (Note that $\fv$ is a random vector that is a function of $\xv$.)
We have $\E_{\xv}[\sE(\fv) - \sE(\fv^*)] = \E_\xv[\gap(\fv, \xv^*)]$. 

For a fixed choice of $\xv$, consider running the algorithm for one more iteration starting from $\xv$ to obtain a vector $\xv'$. Then we have $\E[\sB(\xv')] - \sB(\xv) = \frac1\tau \gap(\fv, \xv)$. This implies $\sB(\xv^*) - \sB(\xv) \geq \frac1\tau \gap(\fv, \xv)$. Taking expectations with respect to $\xv$, we get $\E_{\xv}[\sB(\xv^*) - \sB(\xv)] \geq \frac1\tau\E_{\xv}[ \gap(\fv, \xv)]$.
Thus,
\begin{align*}
    \E_{\xv}[\sE(\fv) - \sE(\fv^*)] &=  \E_{\xv}[\gap(\fv, \xv^*)] \\
     &=   \E_{\xv}[\gap(\fv, \xv) - (\sB(\xv^*) - \sB(\xv))]  \\
    &\leq  (\tau-1)\E_\xv[\sB(\xv^*) - \sB(\xv)] \\
    &\leq   \tau\E_\xv[\sB(\xv^*) - \sB(\xv)] \\
    &\leq  \epsilon\sB(\xv^*) \\
    &= \epsilon \sE(\fv^*).
\end{align*}
\end{proof}


    
%
    
\section{Omitted Proofs from Section \ref{sec:ds}}

\reduction*

\begin{proof}
Given an $n \times n$ Boolean matrix $\mathbf{M}$, we create the following \textit{TreeFlow} data structure.
The graph contains $2n+1$ nodes, namely a special node $x$,  one node $c_j$ for each column $j$ with
$1 \le j \le n$ and one node $d_i$ for each row $i$ with $1 \le i \le n$. There is  an edge $(d_i,c_j)$ if entry $M_{ij}$ = 1. Additionally, every node $c_j$ and every node 
$d_i$ has an edge to $x$. These edges are added to guarantee that the graph is connected. 
We set $r(c,d) = 1$ for every edge $(c,d)$ and
denote this graph by $G$. Let $T$ be the spanning tree of $G$ that is rooted at $x$ and consists of all the edges incident to $x$. Note that the subtree of $T$ rooted at any node $y \ne x$ consists of a single node $y$.

Now consider the sequence of $n$ vector pairs $(\mathbf{u}_t, \mathbf{v}_t)$ of the OuMv problem. Let $(\mathbf{u},\mathbf{v})$ be any such pair. We show below how to compute $\mathbf{u^\top} \mathbf{M} \mathbf{v}$ with $O(n)$ operations in the \textit{TreeFlow} data structure. 
Thus the sequence of $n$ vector pairs leads to $O(n^2)$ operations.
It then follows from the OMv conjecture and Lemma~\ref{conj:omv} that this sequence of $O(n^2)$ operations in the
\textit{TreeFlow} data structure cannot take time $O(n^{3-\epsilon})$, i.e., that it is not
possible that the complexity of both the $\textsf{addvalue}$ operation
and the $\textsf{findflow}$ operation are
$O(n^{1-\epsilon})$.

It remains to show how to compute $\mathbf{u^\top} \mathbf{M} \mathbf{v}$ with $O(n)$ operations in the \textit{TreeFlow} data structure. 
Initially the value $\mathsf{value}(v)$ of all nodes $v$ is 0. Let $Z$ be a large enough constant that we will specify later. First, increase the value of  all nodes to $Z$ by calling $\mathsf{addvalue}(x, Z)$. 

When given $(\mathbf{u},\mathbf{v})$ we decrease the value of each row node $d_i$ with $u_i = 1$ to 0  by 
calling $\textsf{addvalue}(d_i, -Z).$
Then, we perform a $\textsf{findflow}(c_j)$ operation for each column node $c_j$ with $v_j = 1$.
Afterwards we undo these operations again, so that  every node has value Z again. (Alternatively, we could also increase the value of every node $c_j$ with $v_j=1$ and every node $d_i$ with $u_i = 0$ to $2Z$, and instead of undoing the operations, we increase after the query the value of every node to $2Z$.  In which way we never execute an
$\textsf{addvalue}$ operation with negative second parameter.)

Note that $\mathbf{u^\top} \mathbf{M} \mathbf{v} = 1$ iff there exists an edge between a column node $c_j$ with $v_j = 1$ (i.e.  $\mathsf{value}(c_j)=Z$) and a row node $d_i$ with $u_i = 1$ (i.e.  $\mathsf{value}(d_i)=0$). 

We now show that $\mathbf{u^\top} \mathbf{M} \mathbf{v} = 1$ iff $f(c_j)>0$ for some column node $c_j$ with $v_j = 1$.  
(a) Assume first that  $\mathbf{u^\top} \mathbf{M} \mathbf{v} = 1$ and let $c^*$ denote a node $c_j$ and $d^*$ denote a node $d_i$ such that $v_j = 1$, $u_i = 1$ and $M_{ij} = 1.$ We will show that $f(c^*) > 0$.
Recall that the subtree of $c^*$ consists only of $c^*$. The edge $(c^*, d^*)$ leaves the subtree of $c^*$, contributing a positive amount to $f(c^*)$ because $\mathsf{value}(c^*)=Z$ and $\mathsf{value}(d^*)=0$.
All other edges leaving the subtree of $c^*$ contribute a non-negative amount to $f(c^*)$, since $\mathsf{value}(c^*)=Z$ and $\mathsf{value}(d_k)$ for other $k \ne i$ is either $Z$ or 0. Thus $f(c^*) > 0$.  
(b)
Assume next that  $\mathbf{u^\top} \mathbf{M} \mathbf{v} = 0$. In this case every node $c_j$ with $u_j = 1$ (and value $Z$) only has edges to nodes $d_i$ with $v_i = 0$ (and value $Z$). As before the subtree of every node $c_j$ only consists of $c_j$ and, thus, all edges leaving the subtree of $c_j$ contribute 0 to the flow out of the subtree. Thus, 
for every node $c_j$ with $u_j = 1$ we have $f(c_j) = 0$. 

To summarize  we have shown above that $\mathbf{u}\trans \mathbf{M} \mathbf{v} = 1$ iff $f(c_j) > 0$ for some column node $c_j$ with $\mathsf{value}(c_j) = Z$. We will now show how to use the results of the \textsf{findflow} queries returned by an $\alpha$-approximate \emph{TreeFlow} data structure to determine if $f(c_j)$ is positive or zero.  

Here is where we will choose the value of $Z$. The idea is to make $Z$ large enough so that if $f(c_j) > 0$, then $f(c_j)$ is very large. The idea is that this will allow us to distinguish between $f(c_j) = 0$ versus $f(c_j) > 0$, even if we only have access to an $\alpha$-approximation of $S(c_j) - f(c_j) = b(c_j) - f(c_j)$.

It will suffice to choose $Z$ large enough so that if $f(c_j) > 0$, then $f(c_j) > \max\{b(c_j), b(c_j)(1 - \alpha^2)\}$ (As $(1-\alpha^2) < 0$, the second term makes sense if $b(c_j) < 0$.) The value of $Z$ depends on $\alpha$, the supplies $\mathbf{b}$, and the resistances $\mathbf{r}$. For instance, it suffices to choose $Z > \norm{r}_\infty \norm{b}_\infty \alpha^2$. For this choice of $Z$, we have that if $f(c_j) > 0$ then (since it must have an edge to some $d_i$ with $\mathsf{value}(d_i) = 0$), 
$$f(c_j) \geq \frac{\mathsf{value}(c_j) - \mathsf{value}(d_i)}{r(c_j, d_i)} = \frac{Z-0}{r(c_j, d_i)} > \abs{b(c_j)}\alpha^2 > \max\{b(c_j),\, b(c_j)(1 - \alpha^2)\}.$$
Having chosen $Z$ this way, we have the following:
\begin{itemize}
    \item If $b(c_j) \geq 0$, then  $b(c_j) - f(c_j)$ is non-negative if $f(c_j) = 0$, and negative otherwise (because $f(c_j) > b(c_j)$ when $f(c_j) > 0$.) Any $\alpha$-approximation of $b(c_j) - f(c_j)$ allows us to correctly deduce the sign of $b(c_j) - f(c_j)$, hence also whether $b(c_j) -f(c_j) \ge 0$ or
    whether $b(c_j) - f(c_j) < 0$. From this we can deduce whether $f(c_j) = 0$ or $f(c_j) > 0$.
    \item Suppose $b(c_j) < 0$. If $f(c_j) = 0$, the approximate data structure returns an answer in the interval $[b(c_j)\cdot \alpha,\,\frac{ b(c_j)}{\alpha}]$. If $f(c_j) > 0$, it returns an answer in the interval $[(b(c_j) - f(c_j))\cdot \alpha,\, \frac{b(c_j) - f(c_j)}{\alpha}]$. Note that 
    the left endpoint of the first interval is to the right of the right endpoint of the second interval as 
    $f(c_j) > b(c_j)\left(1 - \alpha^2\right)$ implies that $$ \implies b(c_j)\cdot \alpha > \frac{b(c_j) - f(c_j)}{\alpha}.$$
    Since the two intervals for $f(c_j) = 0$ and $f(c_j) > 0$ do not overlap, we can correctly distinguish the two cases using the approximate data structure. 
\end{itemize}

To summarize, each \textsf{findflow} query on $c_j$ allows us to determine if $f(c_j) > 0$ or $f(c_j) = 0$. If the flow is positive for some $c_j$, then the answer is $\mathbf{u^\top} \mathbf{M} \mathbf{v} = 1$, otherwise it is 0.
Note that it requires $O(n)$ operations in the \textit{TreeFlow} data structure to determine one $\mathbf{u^\top} \mathbf{M} \mathbf{v}$ value, which completes the proof.
\end{proof}

\begin{remark}
Note that the proof can be modified to be more similar to the update sequence generated by \ALGNAME\ which alternates between \textsf{addvalue} and \textsf{findflow} operations by inserting
 after each \textsf{addvalue} operation  a \textsf{findflow} operation (whose answer might be ignored for the logic of the proof).
As mentioned, the proof can also be adapted so that the values stored at the nodes are only increased, but this is not necessary for our application.
%
%
\end{remark}

\section{Spectral Approximations}

Let $\Am, \Bm$ be $n \times n$ symmetric, positive semidefinite matrices. We say that $\Bm$ is a $\gamma$-spectral sparsifier of $\Am$ if
\[
\left( 1 - \gamma \right) \Am
\preceq
\Bm
\preceq
\left( 1 + \gamma \right) \Am.
\]


\begin{prop}[Spectral Approximations]
\label{prop:spec_approx}
Suppose $\Am, \Bm \in \mathbb{S}^n_+$ and $\left( 1 - \gamma \right) \Am
\preceq
\Bm
\preceq
\left( 1 + \gamma \right) \Am$. Let $\xv, \yv, \zv, \bv \in \R^n$. Then the following hold:
\begin{enumerate}
    \item $(1-\gamma) \norm{\xv}_{\Am}^2 \leq \norm{\xv}_{\Bm}^2 \leq (1+\gamma)\norm{\xv}_{\Am}^2$
    \item If $\Am\xv = \bv$ and $\Bm\yv = \bv$, then $\norm{\xv - \yv}_{\Am}^2 \leq h(\gamma) \norm{\xv}_{\Am}^2$, where $h(\gamma) = \frac{\gamma^2}{(1-\gamma)^2}$. 
\end{enumerate}
\end{prop}

\begin{proof}
The first one is by definition of $\norm{\xv}_{\Am}^2 = \xv^{\top} \Am \xv$.

For the second one, first we claim that it is sufficient to prove that $\norm{\Am^\dag \bv - \Bm^\dag \bv}_{\Am}^2 \leq h(\gamma) \norm{\Am^\dag \bv}_{\Am}^2$. This is because in general, we have $\xv = \Am^\dag \bv + \uv$, and $\yv = \Bm^\dag \bv + \vv$, for some  $\uv \in \Null(\Am)$ and $\vv \in \Null(\Bm)$. Moreover, the condition 
$
\left( 1 - \gamma \right) \Am
\preceq
\Bm
\preceq
\left( 1 + \gamma \right) \Am
$
implies that $\Null(\Am) = \Null(\Bm)$. Hence, $\norm{\xv - \yv}_{\Am}^2 = \norm{\Am^\dag \bv - \Bm^\dag \bv}_{\Am}^2$, and $\norm{\xv}_{\Am}^2 =  \norm{\Am^\dag \bv}_{\Am}^2$. 

Next, we expand
$\norm{\Am^\dag \bv - \Bm^\dag \bv}_{\Am}^2 \leq h(\gamma) \norm{\Am^\dag \bv}_{\Am}^2$ into 
$$
(\Am^\dag \bv - \Bm^\dag \bv)^T \Am (\Am^\dag \bv - \Bm^\dag \bv) \leq h(\gamma) \bv^T \Am^\dag \bv,
$$
or equivalently,
$$
\bv^T \left(\Am^\dag - \Bm^\dag\right) \Am \left(\Am^\dag - \Bm^\dag\right) \bv \leq h(\gamma)   \bv^T \Am^\dag \bv.
$$
To prove the above inequality, it suffices to prove that 
\begin{equation}
    \label{eq:spec_goal}
     \left(\Am^\dag - \Bm^\dag\right) \Am \left(\Am^\dag - \Bm^\dag\right)  \preceq h(\gamma)  \Am^\dag.
\end{equation}
Multiplying the left and right sides of \Cref{eq:spec_goal} by $\Am^{\frac12}$, we get that (\ref{eq:spec_goal}) is  implied by
\begin{align}
\label{eq:spec_suff}
    \Am^{\frac12}\left(\Am^\dag - \Bm^\dag\right) \Am \left(\Am^\dag - \Bm^\dag\right)\Am^{\frac12}  \preceq h(\gamma) \Am^{\frac12} \Am^\dag \Am^{\frac12}.
\end{align}
Let $\Pi := \Am^{\frac12} \Am^\dag \Am^{\frac12}$ be the projection map onto the row space of $\Am$. Note that $\Pi = \Am^\dag\Am = \Am\Am^\dag$. Also, $\Pi = \Am^{\frac{\dag}{2}}\Am^{\frac{1}{2}} = \Am^{\frac{1}{2}}\Am^{\frac{\dag}{2}}$. These can be seen using the spectral decomposition. 
Now, the reason why \Cref{eq:spec_suff} implies \Cref{eq:spec_goal} is because if we can multiply both sides of (\ref{eq:spec_suff}) with one copy of $ \Am^{\frac{\dag}{2}}$ on the left and one copy of  $\Am^{\frac{\dag}{2}}$ on the right. Then \Cref{eq:spec_suff} becomes  $\Pi\left(\Am^\dag - \Bm^\dag\right) \Am \left(\Am^\dag - \Bm^\dag\right)\Pi \preceq h(\gamma) \Pi\Am^\dag \Pi.$ We have $\Pi (\Am^{\dag} - \Bm^{\dag}) = (\Am^{\dag} - \Bm^{\dag})\Pi  = \Am^{\dag} - \Bm^{\dag}$ because $\Am$ and $\Bm$ have the same null space. Similarly, $\Pi \Am^{\dag}  = \Am^{\dag} \Pi  = \Am^{\dag}$.

To prove (\ref{eq:spec_suff}), first rewrite it as
\begin{align*}
    \left(\Am^{\frac12}\left(\Am^\dag - \Bm^\dag\right) \Am^{\frac12} \right)^2   \preceq h(\gamma) \Pi,
\end{align*}
or equivalently
\begin{align*}
    \left(\Pi - \Am^{\frac12}\Bm^\dag\Am^\frac12 \right)^2   \preceq h(\gamma) \Pi.
\end{align*}
From the spectral approximation
$\left( 1 - \gamma \right) \Am
\preceq
\Bm
\preceq
\left( 1 + \gamma \right) \Am$, we deduce that
$$
\frac{1}{1 + \gamma} \Am^\dag
\preceq
\Bm^\dag
\preceq
\frac{1}{1-\gamma} \Am^\dag,
$$
which, when multiplying on the left and right by $\Am^{\frac12}$, implies that
$$
\frac{1}{1 + \gamma} \Pi
\preceq
\Am^{\frac12}\Bm^\dag\Am^{\frac12}
\preceq
\frac{1}{1-\gamma} \Pi.
$$
This in turn gives
$$
\frac{-\gamma}{1 + \gamma} \Pi
\preceq
\Pi - \Am^{\frac12}\Bm^\dag\Am^{\frac12}
\preceq
\frac{\gamma}{1-\gamma} \Pi.
$$
Observe that any eigenvector of $\Pi - \Am^{\frac12}\Bm^\dag\Am^{\frac12}$ is also an eigenvector of $\Pi$ (they share eigenspaces because $\Am$ and $\Bm$ have the same null spaces).
Moreover, the eigenvalues of $\Pi$ are 0 or 1. 
This implies that the eigenvalues of $\Pi - \Am^{\frac12}\Bm^\dag\Am^\frac12$ are all between $\frac{-\gamma}{1 + \gamma}$ and $\frac{\gamma}{1-\gamma}$. Hence, the eigenvalues of  $\left(\Pi - \Am^{\frac12}\Bm^\dag\Am^\frac12 \right)^2$ are all between 0 and  $ \frac{\gamma^2}{(1-\gamma)^2}$, and thus $\left(\Pi - \Am^{\frac12}\Bm^\dag\Am^\frac12 \right)^2   \preceq \frac{\gamma^2}{(1-\gamma)^2} \Pi$.

\end{proof}

\end{document}